\documentclass[a4paper,USenglish,cleveref]{compgeomtop_submission}

\usepackage[utf8]{inputenc}
\usepackage{amsmath,amssymb,amsfonts}
\usepackage{latexsym}
\usepackage{xspace}
\usepackage{verbatim}
\usepackage{mathtools}
\usepackage[disableredefinitions]{complexity}
\usepackage{xcolor}
\definecolor{defblue}{rgb}{0.121,0.47,0.705}
\definecolor{linkblue}{rgb}{0.082,0.310,0.537}

\usepackage{circledtext}
\circledtextset{resize=real}
\newcommand{\mynumber}[1]{\raisebox{-.3ex}{\circledtext[height=2.2ex,charshrink=0.6]{#1}}}

\hypersetup{colorlinks=true,
	linkcolor=linkblue,
	anchorcolor=linkblue,
	citecolor=linkblue,
	filecolor=linkblue,
	menucolor=linkblue,
	urlcolor=linkblue,
	bookmarksopen=true,
	bookmarksopenlevel=2,
	bookmarksnumbered=true,
	plainpages=false,
}

\let\emph\relax
\DeclareTextFontCommand{\emph}{\color{defblue}\em}
\usepackage{graphicx}
\graphicspath{{./figures/}}
\usepackage{url}
\newtheorem{property}[theorem]{Property}
\Crefname{figure}{Figure}{Figures}
\crefname{proposition}{Proposition}{Propositions}
\Crefname{observation}{Observation}{Observations} %
\crefname{observation}{Observation}{Observations} %
\crefname{property}{Property}{Properties}
\crefname{claim}{Claim}{Claims}
\Crefname{claim}{Claim}{Claims}

\AddToHook{env/definition/begin}{\crefalias{theorem}{definition}}
\AddToHook{env/lemma/begin}{\crefalias{theorem}{lemma}}
\AddToHook{env/example/begin}{\crefalias{theorem}{example}}
\AddToHook{env/property/begin}{\crefalias{theorem}{property}}

\newcommand{\fix}{fix\xspace}
\newcommand{\flex}{flex\xspace}
\newcommand{\true}{{\em true}\xspace}
\newcommand{\false}{{\em false}\xspace}
\newcommand{\forestStory}{\textsc{Forest Storyplan}\xspace}
\newcommand{\outerStory}{\textsc{Outerplanar Storyplan}\xspace}
\newcommand{\planarStory}{\textsc{Planar Storyplan}\xspace}

\title{Outerplanar and Forest Storyplans}
\titlerunning{Outerplanar and Forest Storyplans}

\author{Ji\v{r}\'{i} Fiala}{Charles University, Prague, Czech Republic}{fiala@kam.mff.cuni.cz}%
{https://orcid.org/0000-0002-8108-567X}{}
\author{Oksana Firman}{Institut für Informatik, Universität Würzburg,
	Germany}{oksana.firman@uni-wuerzburg.de}%
{https://orcid.org/0000-0002-9450-7640}{}
\author{Giuseppe Liotta}{Universit\`a degli Studi di Perugia, Perugia, Italy}{giuseppe.liotta@unipg.it}%
{https://orcid.org/0000-0002-2886-9694}{Supported by MUR PON project ARS01 00540 and by MUR PRIN project 2022TS4Y3N.}
\author{Alexander Wolff}{Institut für Informatik, Universität Würzburg, Germany \and \url{https://www.informatik.uni-wuerzburg.de/en/algo/team/wolff-alexander/}}{}%
{https://orcid.org/0000-0001-5872-718X}{}
\author{Johannes Zink}{TUM School of Computation, Information and Technology, Technische Universität München, Germany}{johannes.zink@tum.de}%
{https://orcid.org/0000-0002-7398-718X}{}

\authorrunning{J.~Fiala, O.~Firman, G.~Liotta, A.~Wolff, and J.~Zink} 

\Copyright{Ji\v{r}\'{i}~Fiala, Oksana~Firman, Giuseppe~Liotta, Alexander~Wolff, and Johannes~Zink}


\relatedversion{A preliminary version of this work has appeared in
	Proc.\ 49th International Conference on Current Trends in Theory \& Practice of Computer Science (SOFSEM’24), Springer-Verlag, 2024~\cite{fflwz-ofs-SOFSEM24}.}

\Volume{}
\ArticleNo{1}

\nolinenumbers

\begin{document}
	\maketitle

        \pdfbookmark[1]{Abstract}{Abstract} 
        \begin{abstract}
	We study the problem of gradually representing a complex graph as a
	sequence of drawings of small subgraphs whose union is the complex
	graph.  The sequence of drawings is called \emph{storyplan}, and
	each drawing in the sequence is called a \emph{frame}.
	In an (outer)planar storyplan, every frame is (outer)planar;
        in a forest storyplan, every frame is acyclic. 
        Binucci, Di Giacomo, Lenhart, Liotta, Montecchiani, N{\"o}llenburg,
        and Symvonis [JCSS 2024] proved that every graph of
        treewidth at most~3 admits a planar storyplan and that deciding
        whether a given graph admits a \emph{planar} storyplan is NP-complete.
        They also presented two \FPT\ algorithms, one parameterized with respect
        to the vertex cover number and one with respect to the feedback edge
        set number of the input graph.
	\par
    We first prove that deciding whether a given graph admits an outerplanar storyplan (or a forest storyplan) is \NP-complete.
    Then, we show that the \FPT\ algorithms of Binucci et al.\
    also work for our problem variants with small modifications.
	We identify graph families that admit outerplanar and forest storyplans and graph families
	for which such storyplans do not always exist.  In the affirmative
	case, we present efficient algorithms that produce storyplans with straight-line edges.
\end{abstract}
	
\section{Introduction}
\label{sec:introduction}

A possible approach to the visual exploration of large and complex
networks is to gradually display them by showing a sequence of
\emph{frames}, where each frame contains the drawing of a portion of
the graph.  When going from one frame to the next, some vertices and
edges appear while others disappear. To preserve the mental map, the
geometric representation of vertices and edges that are shared by two
consecutive frames must remain the same.  Informally speaking,
a~\emph{storyplan} for a graph consists of a sequence of frames such
that every vertex and edge of the graph appears in at least one frame.
Moreover, there is a consistency requirement (as for the labels in a
zoomable digital map \cite{BeenNPW08,bnpw-oarcd-CGTA10}): once a vertex
disappears, it may not re-appear.  Hence, after a vertex appears, it
remains visible until all its incident edges are represented; then it
disappears in the transition to the next frame.  See
\cref{fig:peterson} for a storyplan.

\begin{figure}[thb]
  \centering
  \includegraphics{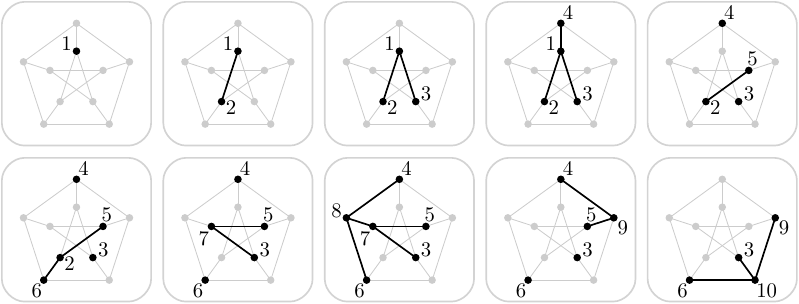}
  \caption{A forest storyplan of the Petersen graph.}
  \label{fig:peterson}
\end{figure}

Since edge crossings are a natural obstacle to the readability of a
graph layout~\cite{DBLP:books/ph/BattistaETT99}, Binucci, Di Giacomo,
Lenhart, Liotta, Montecchiani, N{\"{o}}llenburg, and
Symvonis~\cite{bdllmms-csp-GD22,bdllmms-csp-JCSS24}
introduced and studied the \emph{planar storyplan problem}
that asks whether a given graph admits a storyplan
such that every frame is a crossing-free drawing and in every frame a single
new vertex appears. We call this problem \planarStory.
Binucci et al.\ showed that \planarStory is \NP-complete in general and
fixed-parameter tractable with respect to the vertex
cover number and with respect to the feedback edge set number.
They also proved that every graph of treewidth at
most~3 admits a planar storyplan.

Motivated by the research of Binucci et al., we forward the idea of
representing a graph with a storyplan such that each frame is a
drawing whose visual inspection is as simple as possible.
Specifically, we study the \emph{outerplanar storyplan problem} and
the \emph{forest storyplan problem}, which are defined analogously to
the planar storyplan problem (see \cref{de:storyplan} below).
We call the problems \outerStory and \forestStory, respectively.
We let the classes of graphs that admit planar, outerplanar and
forest storyplans be denoted by $\mathcal{G}_\mathrm{planar}$,
$\mathcal{G}_\mathrm{outerpl}$, and $\mathcal{G}_\mathrm{forest}$,
respectively.  Clearly, $\mathcal{G}_\mathrm{forest} \subseteq
\mathcal{G}_\mathrm{outerpl} \subseteq \mathcal{G}_\mathrm{planar}
\subseteq \mathcal{G}$, where $\mathcal{G}$ is the class of all
graphs. To
further simplify visual inspection, our algorithms draw all frames
with straight-line edges.  Following Dobler, Holzmüller, and Nöllenburg~\cite{dhn-gmps-GD25}, we call storyplans with this property
\emph{geometric storyplans}.
They constructed a graph that admits a planar storyplan but no geometric planar storyplan.
By modifying the \NP-hardness proof of Binucci et al., Dobler et al.\ showed that the geometric version of \planarStory is also \NP-hard.

Beside the work of Binucci et al., our research relates to the graph
drawing literature that assumes either dynamic or streaming models
(see, e.g.,
\cite{DBLP:conf/iv/AbdelaalLHW20,DBLP:journals/ipl/BinucciBBDGPPSZ12,DBLP:journals/jvis/Burch17,LozzoR15,LozzoR19})
and to recent work about graph stories (see, e.g.,
\cite{ddggopt-spssgs-GD22,BattistaDGGOPT23,BorrazzoLFP19,BorrazzoLBFP20}).
The key difference to our work is that these papers
(except~\cite{bdllmms-csp-JCSS24}) assume that the order of the vertices
is given as part of the input.  We now summarize our contribution,
using \emph{$\triangle$-free} as shorthand for triangle-free.  

\begin{itemize}
	\item We establish the chain of strict containment relations
	$\mathcal{G}_\mathrm{forest} \subsetneq \mathcal{G}_\mathrm{outerpl}
	\subsetneq \mathcal{G}_\mathrm{planar} \subsetneq \mathcal{G}$ (see
	\cref{fig:diagram}) by showing that
	\begin{itemize}
		\item there is a $\triangle$-free 6-regular graph that does not admit a
		planar storyplan;
		\item there is a $K_4$-free 4-regular planar graph that (trivially)
		admits a planar storyplan, but does not admit an outerplanar
		storyplan; and
		\item there is a $\triangle$-free 4-regular (nonplanar) graph that
		admits an outerplanar storyplan, but does not admit a forest
		storyplan.
	\end{itemize}
	\item Recall that a \emph{triangulation} is a maximal planar graph; it
	admits a planar drawing where every face (including the outer one) is a triangle. 
    We show that no triangulation (except for~$K_3$) admits an outerplanar storyplan; see \cref{sec:negative}.
    \item We show that both problems \outerStory and \forestStory
	are \NP-complete; \cref{sec:complexity}.
    \item Further we modify the \FPT\ algorithms
	of Binucci et al.~\cite{bdllmms-csp-JCSS24} for \planarStory in
	order to obtain \FPT\ algorithms for \outerStory	and for
	\forestStory; see \cref{sec:fpt}.
    \item On the positive side, we exhibit two graph classes that
        admit geometric outerplanar storyplans. These are subcubic
        graphs (except~$K_4$) and partial 2-trees~-- the latter is
        the class of graphs of treewidth at most~2; it contains all
        outerplanar graphs as a proper subclass.  Both constructions
        can be performed in linear time; see \cref{sec:outer}.  For
        subcubic graphs, every frame of the resulting storyplans
        contains at most five~edges.
      \item A graph must be $\triangle$-free in order to admit a forest
	storyplan.  We show that $\triangle$-free subcubic graphs (as the
	Petersen graph in \cref{fig:peterson}), and $\triangle$-free planar
	graphs admit geometric forest storyplans (which we can compute
	in linear and polynomial time, respectively); see \cref{sec:forest}.
\end{itemize}

\begin{figure}[tb]
	\centering
	\includegraphics{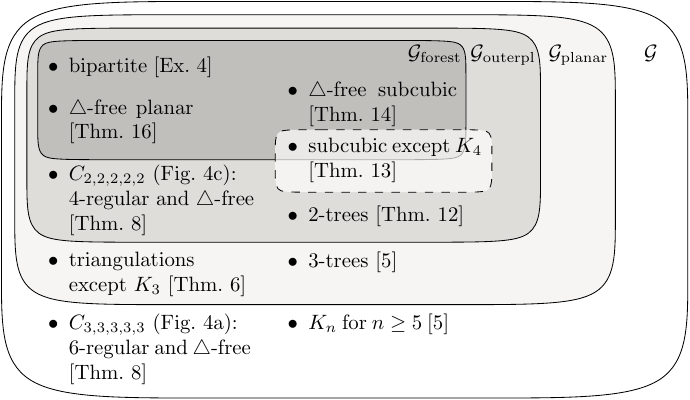}
	\caption{The relationship $\mathcal{G}_\mathrm{forest} \subsetneq
        \mathcal{G}_\mathrm{outerpl} \subsetneq
        \mathcal{G}_\mathrm{planar} \subsetneq \mathcal{G}$ is strict:
        we list graphs and graph classes that
        are contained in one class but not in the other.
        For example, all triangulations admit a planar storyplan
        but not an outerplanar storyplan with the exception of $K_3$,
        which, like all 2-trees, admits an outerplanar storyplan.
        Note that subcubic graphs except $K_4$ admit
        an outerplanar storyplan and some also a forest storyplan.
        Existing results are marked with reference~\cite{bdllmms-csp-JCSS24},
        while new results show the name of the theorem/example.
        For the definition of $C_{2,2,2,2,2}$ and $C_{3,3,3,3,3}$,
        see \cref{thm:negative}.}
	\label{fig:diagram}
\end{figure}

We start with some preliminaries in \cref{sec:preli} and close with
open problems in \cref{sec:open}.

\section{Preliminaries}
\label{sec:preli}

Given positive
integers $m, n$, we use \emph{$[m, n]$} as shorthand for $\{m, m+1,\dots,n\}$,
and we use \emph{$[n]$} as shorthand for $[1, n]$.
Given a graph~$G$, we let \emph{$V(G)$} denote the vertex set of~$G$,
and we let \emph{$E(G)$} $\subseteq {V(G) \choose 2}$ denote the edge set of~$G$.

Our definitions of a planar, an outerplanar, and a forest storyplan
are based on the definition of a planar storyplan of Binucci et
al.~\cite{bdllmms-csp-JCSS24}.

\begin{definition}
  \label{de:storyplan}
	A \emph{planar storyplan}
	$\mathcal{S}=\langle \alpha, (D_i)_{i\in [n]} \rangle$ of $G$ is a
	pair defined as follows.  The first element is a bijection
	$\alpha \colon V(G) \rightarrow [n]$ that represents a total order of the
	vertices of $G$.  For a vertex $v \in V(G)$, 
    let	$\omega(v) = \max_{u \in N[v]} \alpha(u)$, where $N[v]$, the \emph{closed neighborhood} of~$v$, is the set
	containing $v$ and its neighbors.  The interval $[\alpha(v),\omega(v)]$ is the
	\emph{lifespan} of~$v$.  We say that $v$ \emph{appears} at step
	$\alpha(v)$, is \emph{visible} at step~$i$ for each $i \in [\alpha(v),\omega(v)]$, and
	\emph{disappears} at step $\omega(v)+1$.
        Note that a vertex disappears
	only when all its neighbors have appeared.  The second element of
	$\mathcal{S}$ is a sequence of drawings $(D_i)_{i\in [n]}$, called
	\emph{frames} of~$\mathcal{S}$, such that, for $i \in [n]$:
	(i)~$D_i$ is a drawing of the graph $G_i$ induced by the vertices
	visible at step~$i$, (ii)~$D_i$ is planar, (iii)~the point
	representing a vertex $v$ is the same over all drawings that contain
	$v$, and (iv)~the curve representing an edge $e$ is the same over
	all drawings that contain~$e$.
\end{definition}

We emphasize that though in our definition of storyplans we allow
that edges can be represented by curves, our constructions use only
straight-line edges.
For an \emph{outerplanar storyplan} and a \emph{forest storyplan}, we
strengthen requirement~(ii) to $D_i$ being outerplanar and $D_i$ being
a crossing-free drawing of a forest, respectively.  In what follows, we will sometimes use a slight
variant of \cref{de:storyplan}, in which we enrich the sequence
$(D_i)_{i\in [n]}$ of frames by explicitly representing the portions
of the drawings that consecutive frames have in common.  More
precisely, for $i \in [n-1]$, let $D_i'=D_i \cap D_{i+1}$.  Then, a
storyplan is a sequence of drawings
$\langle D_1, D'_1, \dots, D_{n-1}, D'_{n-1}, D_n\rangle$, where in
each step~$i < n$, we first introduce a vertex (in~$D_i$) and then
remove all \emph{completed} vertices (in~$D_i'$), that is, the
vertices that disappear in the next step.  For $i \in [n-1]$,
we define (similarly to~$D_i'$) $G_i'$ as the graph induced by the vertices
of the graphs~$G_i$ and~$G_{i+1}$.  We now list some useful observations.

\begin{property}\label{prop:partial-graph}
	If a graph $G$ admits a planar, an outerplanar, or a forest
	storyplan, then the same holds for any subgraph of~$G$.
	Conversely, if a graph $G$ does not admit a planar, an outerplanar,
	or a forest storyplan, then the same holds for all supergraphs of $G$.
\end{property}

\begin{lemma}[\cite{bdllmms-csp-JCSS24}]\label{le:bipartite-3}
	Let $K_{a,b} = (A \cup B, E)$ be a complete bipartite graph with
	$a = |A|$, $b=|B|$, and $3 \le a \le b$.  Let
	$\mathcal{S}=\langle \alpha, \{D_i\}_{i\in [a+b]} \rangle$ be a planar
	storyplan of $K_{a,b}$.  Exactly one of $A$ and $B$ is such that all
	its vertices are visible in some frame $i \in [a+b]$.
\end{lemma}

\begin{example}[Bipartite graphs]%
	\label{rem:bipartite}
	Every bipartite graph admits a forest storyplan: first add all
	vertices of one set of the bipartition and then, one by one, the
	vertices of the other set.  Note that each vertex of the second set
	is visible in only one frame.
\end{example}

Now we modify \cref{le:bipartite-3} in order to obtain a
property for outerplanar (forest) storyplans of complete
bipartite graphs. The proof of \cref{le:bipartite-small}
is analogous to the proof of \cref{le:bipartite-3}
(see \cite{bdllmms-csp-JCSS24}).

\begin{lemma}
	\label{le:bipartite-small}
	Let $K_{a,b} = (A \cup B, E)$ be a complete bipartite graph with
	$a = |A|$, $b=|B|$, and $3 \le a \le b$ ($2 \le a \le b$).  Let
	$\mathcal{S}=\langle \alpha, \{D_i\}_{i\in [a+b]} \rangle$ be an
	outerplanar (forest) storyplan of $K_{a,b}$.
	One of partition $A$ and $B$ is such that all
	its vertices are visible at some $i \in [a+b]$ and no two
	vertices from the other partition are visible at any step
	$j \in [a+b]$.
\end{lemma}

\begin{proof}
	First we show that there exists a frame, where $A$ or $B$ is
	completely visible. Recall that $\alpha(v) \in [n]$ is the step
	where a	vertex $v$ appears and $\omega(v) \in [n]$ with $\alpha(v) \le \omega(v)$
	is the step where $v$ disappears.
	Let $i$ be such that $D_i$ contains the largest number $t$ of
	vertices of $A$ over all frames of $\mathcal{S}$. 
	If $t=a$, we are done. 
	If $t<a$, there exist two vertices $u,v$ of $A$ such that
	$\omega(u) < \alpha(v)$. Note that all vertices in $B$ are adjacent to $u$,
	and hence they all appear at some step smaller than or equal to
	$\omega(u)$. On the other hand,  since all vertices in $B$ are
	adjacent to $v$ as well, they cannot disappear before $\alpha(v)+1$.
	It follows that all vertices of $B$ are visible at step $\omega(u)$.
	
	Now we show that no two vertices from one of the partitions
	$A$ and $B$ are visible at any frame $j \in [a+b]$.
	Consider the interval $I=[s,t] \subseteq [a+b]$ of maximal
	length such that the vertices of $A$ or $B$, say $A$, are all
	visible. As we have already shown, $I$ is not empty.
	Let $k$ be one of the steps that contains the largest number $b'$
	of vertices of $B$.
        Observe that, since $I$ is maximal, one
	vertex of $A$ appears at $s$. Therefore, any vertex of $B$
	visible at a step smaller than $s$ is visible also at step $s$.
	Similarly, by the maximality of $I$, one vertex of $A$
	disappears at step $t+1$, therefore any vertex of $B$ visible at
	a step greater than $t$ is visible also at step $t$.
	Consequently, we can assume that $k \in I$, and we can conclude that
	$b' \le 1$ since it is not possible to have three (two) vertices
	from one partition and two vertices from the other partition at
	any	frame of the outerplanar (forest) storyplan $\mathcal{S}$.
\end{proof}

Note that if one of the partition $A$ or $B$ has exactly two
vertices, at most two vertices of the other partition are visible
at some frame of an outerplanar storyplan.

\section{Separation of Graph Classes}
\label{sec:negative}

Trivially, triangulations admit planar storyplans, but as we show now,
no triangulation (except for $K_3$) admits an outerplanar storyplan.

\begin{theorem}
	\label{thm:triangulations}
	No triangulation (except for $K_3$) admits an outerplanar storyplan.
\end{theorem}

\begin{proof}
    For a triangulation, the closed neighborhood of each vertex induces
	a wheel, which is not outerplanar. For the first vertex that
	disappears according to a given storyplan, however, its whole closed
	neighborhood, which is not outerplanar, must be visible.
\end{proof}

\begin{example}[Platonic graphs]
	\label{rem:platonic}
	According to \cref{thm:triangulations}, the tetrahedron, the
	octahedron, and the icosahedron do not admit outerplanar storyplans
	because they are triangulations.  The cube is bipartite; hence, it
	admits a forest storyplan due to \cref{rem:bipartite}.
    The	dodecahedron admits a forest storyplan, see \cref{fig:dodecahedron}. 
    Indeed, the dodecahedron is $\triangle$-free and cubic, and we will show later in \cref{thm:subcubic} that these conditions are sufficient for the existence of a forest storyplan. 
\end{example}

\begin{figure}[bth]
	\centering
	\includegraphics{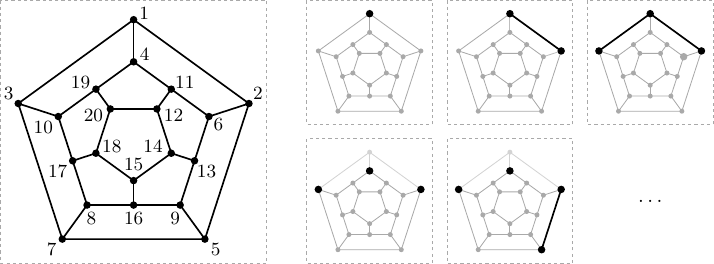}
	\caption{The dodecahedron with a vertex numbering (left) that
          corresponds to a forest storyplan (started to the right).}
	\label{fig:dodecahedron}
\end{figure}

We now separate the graph classes $\mathcal{G}_\mathrm{forest}$,
$\mathcal{G}_\mathrm{outerpl}$, $\mathcal{G}_\mathrm{planar}$, and
$\mathcal{G}$; see \cref{fig:diagram}.

\begin{theorem}
	\label{thm:negative}
	The following statements hold:
	\begin{enumerate}
		\item There is a $\triangle$-free 6-regular graph that does not admit a
		planar storyplan; hence
		$\mathcal{G}_\mathrm{planar} \subsetneq \mathcal{G}$.
		\item There is a $K_4$-free 4-regular planar graph that does not
		admit an outerplanar storyplan; hence
		$\mathcal{G}_\mathrm{outerpl} \subsetneq
		\mathcal{G}_\mathrm{planar}$.
		\item There is a $\triangle$-free 4-regular (nonplanar) graph that
		admits an outerplanar storyplan, but does not admit a forest
		storyplan; hence
		$\mathcal{G}_\mathrm{forest} \subsetneq
		\mathcal{G}_\mathrm{outerpl}$.
	\end{enumerate}
\end{theorem}

\begin{figure}[tb]
	\centering
	\begin{subfigure}[t]{.3\linewidth}
		\centering
		\includegraphics[page=1]{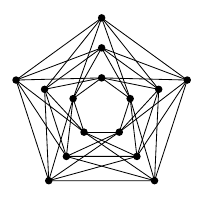}
		\caption{$C_{3,3,3,3,3}$}
		\label{fig:triangle-free1}
	\end{subfigure}
	\hfill
	\begin{subfigure}[t]{.3\linewidth}
		\centering
		\includegraphics{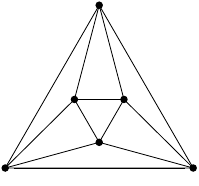}
		\caption{the octahedron graph}
		\label{fig:octahedron}
	\end{subfigure}
	\hfill
	\begin{subfigure}[t]{.3\linewidth}
		\centering
		\includegraphics[page=3]{triangle-free-without-storyplan}
		\caption{$C_{2,2,2,2,2}$}
		\label{fig:triangle-free3}
	\end{subfigure}
    
    \bigskip 
    
	\begin{subfigure}[t]{1\linewidth}
    
\includegraphics[page=4,width=.19\linewidth]{triangle-free-without-storyplan}\hfil%
\includegraphics[page=5,width=.19\linewidth]{triangle-free-without-storyplan}\hfil%
\includegraphics[page=6,width=.19\linewidth]{triangle-free-without-storyplan}\hfil%
\includegraphics[page=7,width=.19\linewidth]{triangle-free-without-storyplan}\hfil%
\includegraphics[page=8,width=.19\linewidth]{triangle-free-without-storyplan}

\vspace{.8ex}

\includegraphics[page=9,width=.19\linewidth]{triangle-free-without-storyplan}\hfil%
\includegraphics[page=10,width=.19\linewidth]{triangle-free-without-storyplan}\hfil%
\includegraphics[page=11,width=.19\linewidth]{triangle-free-without-storyplan}\hfil%
\includegraphics[page=12,width=.19\linewidth]{triangle-free-without-storyplan}\hfil%
\includegraphics[page=13,width=.19\linewidth]{triangle-free-without-storyplan}

        \caption{an outerplanar storyplan for $C_{2,2,2,2,2}$}
		\label{fig:triangle-free3_storyplan}
	\end{subfigure}
	\caption{Three graphs from the proof of \cref{thm:negative}.
		The graph in (\subref{fig:triangle-free1}) is
		$\triangle$-free and does not admit any planar storyplan.
		The octahedron graph in (\subref{fig:octahedron}) does not
		admit any outerplanar storyplan.  The graph in
		(\subref{fig:triangle-free3}) is $\triangle$-free and does
		not admit any forest storyplan, (\subref{fig:triangle-free3_storyplan}) the vertex numbering
		corresponds to an outerplanar storyplan~-- if vertex~8 is
		placed near the position of vertex~6.}
	\label{fig:triangle-free}
\end{figure}

\begin{proof}%
\begin{enumerate}%
\item The graph $C_{3,3,3,3,3}$ (see \cref{fig:triangle-free1}) is
  $\triangle$-free and 6-regular, but does not admit a planar
  storyplan as we will now show.  For brevity, let $G=C_{3,3,3,3,3}$,
  and let $V_1 \cup \dots \cup V_5$ be the partition of the vertex set
  of~$G$ into independent sets of size~3. Note that, for $i \in [5]$,
  $G[V_i \cup V_{(i \bmod 5)+1}]$ is isomorphic to $K_{3,3}$.
  For $K_{3,3} = G[V_1 \cup V_2]$, we know by \cref{le:bipartite-3}
  that, in any planar storyplan, either all vertices of $V_1$ or all
  vertices of $V_2$ are shown simultaneously, say, those of $V_1$.
  This implies, again by \cref{le:bipartite-3}, that the vertices
  of~$V_2$, but also the vertices of~$V_5$ are {\em not} shown
  simultaneously.  This, in turn, means that the vertices of~$V_3$
  {\em are} shown simultaneously, and the vertices of~$V_4$ are also
  shown simultaneously.  Neither the vertices of $V_3$ can disappear
  before the vertices of $V_4$ have appeared and vice versa.  Thus,
  there must be a frame with a drawing of the non-planar graph
  $G[V_3 \cup V_4]=K_{3,3}$.
		
		\item Observe that the octahedron (see \cref{fig:octahedron})
		is planar, 4-regular, and $K_4$-free, but does not admit
		an outerplanar storyplan due to \cref{rem:platonic}.
		
		\item The graph $C_{2,2,2,2,2}$ (see \cref{fig:triangle-free3}) is
		$\triangle$-free and 4-regular, but does not admit a forest storyplan.  
        The proof is analogous to the one for the graph $C_{3,3,3,3,3}$.
        There needs to be a frame with a drawing of $K_{2,2}$, which is not a tree.
		On the other hand, the order of the vertices shown in
		\cref{fig:triangle-free3} yields an outerplanar storyplan.  Note
		that we cannot use the vertex positions exactly as in the figure,
		but if we place vertex~8 near the position of vertex~6 (which will
		have disappeared by then), every frame is crossing-free (see \cref{fig:triangle-free3_storyplan}). \qedhere
\end{enumerate}
\end{proof}

\section{Computational Complexity}
\label{sec:complexity}

In this section, we show that both problems \textsc{Forest Storyplan}
and \textsc{Outerplanar Storyplan} are \NP-hard.  To this end, we
reduce from \textsc{Positive One-In-Three 3-SAT} in such a way that,
if the given 3-SAT formula is a yes-instance, then the graph that we
construct admits a forest storyplan (and hence an outerplanar
storyplan), whereas if the 3-SAT formula is a no-instance, then the
graph that we construct does not even admit an outerplanar storyplan
(let alone a forest storyplan).

In \textsc{One-In-Three 3-SAT} one is given a Boolean formula in
conjunctive normal form with at most three literals per clause and the
task is to decide whether there exists a variable assignment such that
in every clause exactly one literal is assigned to be \true. By
Schaefer's dichotomy theorem \cite{s-csp-STOC78}, \textsc{One-In-Three
3-SAT} is \NP-hard even if no negative literals are admitted;
the problem is then called \textsc{Positive One-In-Three 3-SAT}.

Given an instance $F$ of \textsc{Positive One-In-Three 3-SAT}, we
construct in polynomial time a graph $G_F$ such that $G_F$ admits a
forest storyplan, when $F$ is a yes-instance, and it does not admit even an
outerplanar storyplan, when $F$ is a no-instance.  We may assume that $F$
does not contain a clause with two occurrences of the same variable.
(If there was such a clause of size~2, then $F$ would be a
no-instance.  Otherwise, in a clause of size~3, the variable that
occurs twice would have to be assigned to \false, and we could remove
it from the instance.)

We present our variable and clause gadgets and then show the
correctness of the reduction.  Let $n$ be the number of variables in
$F$, let $m$ be the number of clauses in $F$, and, for
$i \in [n]$, let $p_i$ be the number of occurrences of the
variable $x_i$ in $F$.

\subparagraph{Variable gadget.}

For every $i \in [n]$, we take as gadget for the variable~$x_i$ a
copy~$H_i$ of the complete bipartite graph~$K_{6p_i, 3}$.  Let~$A_i$
and~$B_i$ be the two sets in the bipartition of the vertex set
of~$H_i$ such that $|A_i|= 6p_i$ and $|B_i|=3$; see
\cref{fig:variable-gadget}.  Recall that every bipartite graph admits
a forest storyplan; see \cref{rem:bipartite}.

\begin{figure}[tbh]
	\centering
	\includegraphics{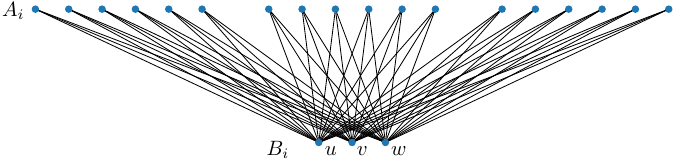}
	\caption{The gadget~$H_i=K_{18,3}$ for a variable~$x_i$ that
          occurs three times in $F$.}
	\label{fig:variable-gadget}
\end{figure}

We know that, for every $i \in [n]$, exactly one of~$A_i$ and~$B_i$ is
completely visible in some frame of an outerplanar storyplan.  This is
due to \cref{le:bipartite-small} and to the fact that $6p_i > 3$.  We
call this set in the bipartition \emph{fix} and the other set
\emph{flex}.  No two vertices of the \flex set are visible
in the same frame.  For every $i \in [n]$, we assign the
variable~$x_i$ to \true if~$A_i$ is \fix, and to \false if~$A_i$
is~\flex.  For each occurrence of~$x_i$ in a clause of~$F$, we use six
unique vertices of~$A_i$ in the corresponding clause gadget.

\subparagraph{Clause gadget.}

We first construct a gadget for a clause that consists of two variables,
then we construct a gadget for a clause that consists of three variables.

Let $C_t = (x_i, x_j)$ with $i,j \in [n]$, $i \ne j$,
and $t\in [m]$ be a clause of~$F$
that consists of two (different) variables (assuming that such a
clause exists).  The gadget for~$C_t$ consists of six vertices
of~$A_i$ and six vertices of~$A_j$ such that the graph induced by
these twelve vertices is the complete bipartite graph~$K_{6,6}$.  Note
that, by \cref{le:bipartite-small}, exactly one of the sets in the
partition of the vertex set of this~$K_{6,6}$ has to be \fix.  If the
\fix~part appears first, then the vertices from the \flex~part appear
one by one, and we obtain a valid forest storyplan for
such a clause gadget.  Note that the \fix~part may stay visible after
the last vertex of the \flex~part disappears. We use this observation
later to show how to synchronize different gadgets.

Let $C_t = (x_i, x_j, x_k)$ with $i,j,k \in [n]$,
$i \ne j \ne k \ne i$, and $t\in [m]$ be a clause of~$F$ that consists
of three different variables.  The gadget~$I_t$ for~$C_t$ is the graph
depicted with black edges inside the yellow disk in
\cref{fig:clause-gadget}.  The vertex set of the graph~$I_t$
is~$A_i \cup A_j \cup A_k$; the three sets (of size~6) are shown in
different colors in \cref{fig:clause-gadget}.  Note that every cycle
of~$I_t$ contains at least one vertex of each color.  The graph~$I_t$
admits a forest storyplan if exactly one of $A_i$, $A_j$,
and $A_k$ is \fix.  We order the vertices as shown in
\cref{fig:clause-gadget}: first, the vertices of the \fix~part appear
(the blue vertices in \cref{fig:clause-gadget}), then the remaining
vertices appear one by one. Note that the \fix~part may stay visible after
the last vertex of the \flex~parts disappears.

\begin{figure}[htb]
	\centering
	\includegraphics[page = 3]{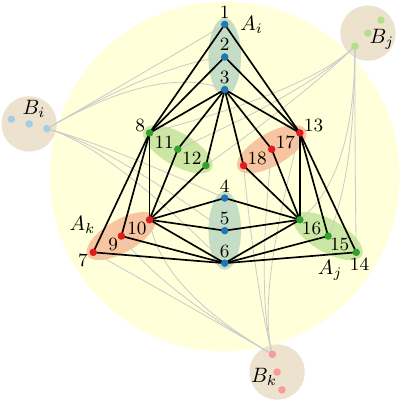}
	\caption{The gadget for the clause $x_i \lor x_j \lor x_k$ (inside the large yellow disk).
          The blue, green, and red vertices are the vertices of~$A_i$,
          $A_j$, and~$A_k$, respectively.  The numbers correspond to
          an ordering of the vertices that yields a forest storyplan; in
          this case, $A_i$ is \fix, and $A_j$ and $A_k$ are \flex.
          For clarity, not all edges (gray) to adjacent vertices (in
          $B_i$, $B_j$, and $B_k$) are shown.}
	\label{fig:clause-gadget}
\end{figure}

\subparagraph{Correctness.}
We first show that if every clause in $F$ contains exactly one \true
variable, then a forest storyplan of $G_F$ exists.  
In order to synchronize different clause gadgets, the vertices of $G_F$ appear in three phases,
in the following order: in the first phase,
for each $i \in[n]$, the \fix~part of the variable gadget~$H_i$ appears.
In the second phase, for each $t \in [m]$, the \flex~parts of the variable
gadgets involved in the clause~$C_t$ appear (and immediately
disappear) in the order described in the paragraph ``Clause gadget''.
In the third phase, for each $i \in [n]$ such that $x_i$ is set to \true,
the vertices in~$B_i$ appear (and immediately disappear) one
after the other (since they are \flex).
Note that this determines an ordering of the vertices of~$G_F$.

We now assume that each clause contains exactly one \true variable.
Note that the vertices of the \fix parts are not connected,
thus in the first phase, we obtain a set of isolated
vertices and hence a set of frames of a valid forest storyplan.
After adding the \flex parts in the second phase, we further obtain
frames of a valid forest storyplan.  To see this,
let $t \in [m]$, and let $h$ be a vertex of one of the \flex parts of
the clause gadget~$I_t$.  As mentioned in the paragraph ``Clause
gadget'' above, the frame in which $h$ (together with its neighbors
from~$I_t$) appears consists of a drawing of a forest.
Apart from its neighbors in~$I_t$, $h$ is
adjacent only to the three vertices of the \fix part of the
corresponding variable gadget, and when $h$ is visible,
these three vertices are only adjacent to~$h$.
Thus, the frame in which $v$ appears
consists of a drawing of a forest.
Finally, each vertex~$u$ of a \flex part ($u \in B_i$ for
some $i \in [n]$) that is added in the third phase is
adjacent to vertices of a \fix part ($A_i$) that are only
adjacent to~$u$ when $u$ is visible.
Note that, since we do not insist on straight-line edges,
a partial drawing of a forest can always be extended without
introducing crossings.
Thus, the corresponding frame consists of a crossing-free
drawing of a forest.

It remains to show that $G_F$ has no outerplanar storyplan if in each variable assignment either zero,
two, or three variables are assigned \true in some clause~$C_t$, $t \in [m]$. We argue that 
the corresponding clause gadget itself or, in the gadget of a
variable~$x_i$ that occurs in~$C_t$, the graph induced by~$B_i$ does
not admit an outerplanar storyplan.

Recall that if a clause consists of only two variables, the
corresponding clause gadget admits a forest storyplan only if
exactly one part of the gadget is \fix~(by
\cref{le:bipartite-small}).
So we now consider a clause $C_t = (x_i, x_j, x_k)$ with
$t \in [m]$ and $i,j,k \in [n]$ that consists of three different
variables.  Recall that~$I_t$ is the clause gadget corresponding
to~$C_t$ and that its vertices are contained in $A_i \cup A_j \cup A_k$.
Let $a,b \in A_i \cap V(I_t)$, $c,d \in A_j \cap V(I_t)$,
and $e,f \in A_k \cap V(I_t)$ be the vertices of degree~6 in~$I_t$.
Furthermore, let $a',a'',b',b'' \in A_i \cap V(I_t)$,
$c',c'',d',d'' \in A_j \cap V(I_t)$, and
$e',e'', f',f'' \in A_k \cap V(I_t)$ be the corresponding vertices
of degree~2 in~$I_t$; see~\cref{fig:clause-gadget-negative-cases}.

\begin{figure}[htb]
	\centering
        \includegraphics[page=3]{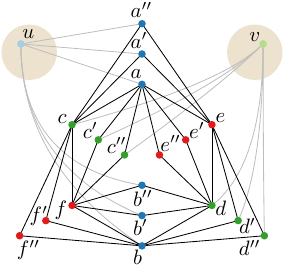}
	\caption{The clause gadget.  The blue, green, and red vertices
          are the vertices of~$A_i$, $A_j$, and~$A_k$, respectively.}
	\label{fig:clause-gadget-negative-cases}
\end{figure}

We distinguish the following three cases.
In each case, we assume that $G_F$ admits an outerplanar storyplan
$\mathcal{S} =\langle \alpha, (D_r)_{r \in [n_F]} \rangle$ (where $n_F$
is the number of vertices of $G_F$).  We show that this leads to a
contradiction.

\smallskip\noindent\textbf{Case 1.}
All parts $A_i$, $A_j$, and $A_k$ are \flex.

Recall that no two vertices of a \flex~part of the same variable gadget
are visible at the same frame
(due to \cref{le:bipartite-small}).
Without loss of generality, assume that $a$ is the first degree-6 vertex
of~$I_t$ that appears in $\mathcal{S}$,
i.e., $\alpha(a) \le \alpha(v)$ for every $v \in \{b,c,d,e,f\}$. 
Since $a$ and $b$ are from the same (\flex) partition $A_i$,
$b$ can appear only after $a$ disappears.  In \cref{fig:dependencies},
this precedence constraint is represented by arc~\mynumber{1}.
Thus, all neighbors of~$a$ (and, hence,~$c$) appear before~$b$
(see arc~\mynumber{2}).
Vertex~$d$ can disappear only after~$b$ has appeared (see \mynumber{3}).
Since~$c$ and $d$ both belong to~$A_j$, they cannot appear simultaneously,
so~$d$ must appear after~$c$ has disappeared (see \mynumber{4}).
We use a similar argument for $e$ and~$f$.
Since $e$ is a neighbor of~$a$,
$e$ appears before $b$ appears (see \mynumber{5}).
Since $f$ is a neighbor of~$b$,
$f$ cannot
disappear before~$b$ has appeared (see \mynumber{6}).
Therefore, $f$ appears after $e$ has disappeared
(see \mynumber{7}) because both belong to~$A_k$.
However, $c$ cannot disappear before its neighbor $f$ has appeared
(see \mynumber{8}) and $e$ cannot disappear before its neighbor~$d$
has appeared (see \mynumber{9}).  Now the circular sequence \mynumber{4},
\mynumber{9}, \mynumber{7}, \mynumber{8} of precedence constraints
constitutes the desired contradiction.

\begin{figure}[htb]
  \centering
  \includegraphics[page=2]{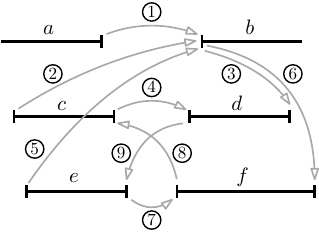}
  \caption{Precedence constraints between the start and end points of
    the visibility intervals of vertices of the clause gadget (in
    Case~1).}
  \label{fig:dependencies}
\end{figure}

\smallskip\noindent\textbf{Case 2.}
All parts $A_i$, $A_j$, and $A_k$ are \fix.

As $K_{2,3}$ is not outerplanar, for each subgraph isomorphic to $K_{2,3}$, it holds that its five vertices do not appear in the same frame.
On the other hand, as in each of the six occurrences of $K_{2,3}$ in $I_t$,
the three vertices of degree~2 form always a subset of either $A_i$, $A_j$, or $A_k$ which all are \fix,
such a $K_{2,3}$ can only be realized
if its two vertices of degree~$3$ do not appear simultaneously.
Therefore, $a$, $d$, and $f$ have pairwise disjoint lifespans, as well as $b$, $c$, and $e$.

Assume without loss of generality that $\omega(a)<\alpha(d)$ (i.e., $a$ disappears before $d$ appears) and 
$\omega(d)<\alpha(f)$.
Consequently, $\omega(b)<\alpha(c)$ and $\omega(c)<\alpha(e)$ because $A_i$ and~$A_j$ are \fix and, therefore, $a$ and $b$ appear in a common frame (analogously for $c,d\in A_j$ and $e,f\in A_k$).

Assume first that $\omega(a) \le \omega(b)$.
Then, since $c$ is a neighbor of $a$,
$\alpha(c) \le \omega(a)$,
and hence $\alpha(c) \le \omega(b)$.
This is a contradiction
to $\omega(b)<\alpha(c)$ as observed above.

Now assume that $\omega(b) \le \omega(a)$.
Then, since $d$ is a neighbor of $b$,
$\alpha(d) \le \omega(b)$,
and hence $\alpha(d) \le \omega(a)$.
This is a contradiction
to $\omega(a)<\alpha(d)$ as observed above.

\smallskip\noindent\textbf{Case 3.}
$A_i$, $A_j$ are \fix, and $A_k$ is \flex.

In order to show that $G_F$ does not yield an outerplanar
storyplan, consider the clause
gadget~$I_t$.  Let~$I'_t$ be a subgraph of $G_F$ induced
by the vertices of~$I_t$ and two additional vertices, a vertex
$u \in B_i$ and a vertex $v \in B_j$. Note that both
partitions $B_i$ and $B_j$ are \flex.
Recall that all vertices of a \fix~part of a variable gadget
are visible at some frame (due to \cref{le:bipartite-small})
and vertices of each \flex~part have disjoint lifespans, in particular also vertices $e$ and $f$. 

Observe that in~$I'_t$,
$K_{3,3}$ appears as the following subgraphs with these vertex partitions:
\begin{enumerate}[(C1)]
    \item $\{a, a', a''\},\{u, c, e\}$
    \item $\{b, b', b''\},\{u, d, f\}$
    \item \label{item:af}
    $\{c, c', c''\},\{v, a, f\}$
    \item \label{item:be}
    $\{d, d', d''\},\{v, b, e\}$
\end{enumerate}
Since $A_i$, $A_j$ are \fix,
this implies that the vertices in
each of the triples
$\{u,c,e\}$, $\{u,d,f\}$, $\{v,a,f\}$, and $\{v,b,e\}$ have pairwise disjoint lifespans.

Assume first that $u$ disappears before $c$ appears, i.e. $\omega(u)<\alpha(c)$.
The path $u,a,c,v$ in $I_t'$ yields $\alpha(a)\le\omega(u)<\alpha(c)\le\omega(a)$ and since the lifespans of $a$ and $v$ are disjoint, but $v$ is adjacent to $c$,
we get $\omega(a)<\alpha(v)$.
Thus, the lifespan of~$u$
strictly precedes the lifespan of~$v$.

Both $e$ and $f$ have disjoint lifespans from the lifespans of $u$ and $v$, so they must appear after $u$ disappears and both disappear before $v$ appears as witnessed by the paths $u,a,e,d,v$ and $u,b,f,c,v$ in $I_t'$. 

As $e$ and $f$ are in $A_k$,
which is \flex,
their lifespans are disjoint.
If $\omega(e)<\alpha(f)$, then since $b$ is a common neighbor of $u$ and $f$, the lifespans of $b$ and $e$ intersect, which contradicts their occurrence in (C\ref{item:be}).
Otherwise, if $\omega(f)<\alpha(e)$, then since $a$ is a common neighbor of $u$ and $e$, the lifespans of $a$ and $f$ intersect,
which contradicts their occurrence in (C\ref{item:af}).
In both cases, we have derived a contradiction to the existence of an outerplanar storyplan.

Assume now that $c$ disappears before $u$ appears, i.e. $\omega(c)<\alpha(u)$.
We will see that we can apply the symmetric argument by switching the roles of the vertices.
The path $v,c,a,u$ in $I_t'$ yields $\alpha(c)\le\omega(v)<\alpha(a)\le\omega(c)$ and since the lifespans of $c$ and $u$ are disjoint, but $u$ is adjacent to $a$,
we get $\omega(c)<\alpha(u)$.
Thus, the lifespan of~$v$
strictly precedes the lifespan of~$u$.

Both $e$ and $f$ have disjoint lifespans from the lifespans of $u$ and $v$.
By the same argument as above,
$e$ and~$f$ appear between $v$ and~$u$.
With a symmetric argument,
we end up with a similar contradiction as above.

\smallskip
We have shown that the graph $G_F$ admits a forest storyplan
if the formula $F$ is satisfiable. Otherwise $G_F$ does not admit
even an outerplanar storyplan.  Observe that the size of~$G_F$
is linear in the size of~$F$:
there is a constant number of vertices
and edges in each clause gadget, the total number of vertices in
sets $B_1,\dots,B_n$ is $3n$, and the total number of
occurrences of all variables is at most $3m$,
which bounds the number of edges
outside the clause gadgets by~$18m$.

As a result we obtain the \NP-hardness of both problems.

\begin{theorem}
    Both \textsc{Forest Storyplan} and \textsc{Outerplanar
        Storyplan} are \NP-hard.
\end{theorem}

In the following, we show that both problems are in~\NP.
The proof is based on the corresponding proof for \textsc{Planar Storyplan} of
Binucci et al.~\cite{bdllmms-csp-JCSS24}.

\begin{theorem}
  Both \textsc{Forest Storyplan} and \textsc{Outerplanar
    Storyplan} are in \NP.
\end{theorem}

\begin{proof}
	In order to show that \textsc{Forest Storyplan} is in \NP, it is
	enough to guess the vertex order and check in each step whether the graph
	is a forest. We do not need to check the drawings in each step
	since we can always extend a drawing of a tree by adding new
	vertices and edges avoiding crossings.
	
	For \textsc{Outerplanar Storyplan} we need to additionally check the
	drawings in each step and their ``consistency''~-- the common parts
	of the drawings in each pair of consecutive steps have to coincide.
	To this end, we rely on the proof of
        \cite[Theorem 3]{bdllmms-csp-JCSS24}, where Binucci
	et al.\ showed that \textsc{Planar Storyplan} is in \NP, but instead
	of looking for planar embeddings, we look for outerplanar ones.
\end{proof}

\section{Parameterized Algorithms}
\label{sec:fpt}

In this section, we observe that the existing parametrized algorithms (with respect to 
two common structural parameters, namely 
vertex cover number and feedback edge set number) for
\planarStory \cite{bdllmms-csp-JCSS24} can be modified to obtain corresponding algorithms for
\forestStory and for \outerStory.

\subsection{\FPT\ Algorithm with Respect to the Vertex Cover Number}

First we show how to obtain \FPT\ algorithms with respect to the vertex
cover number for \forestStory and \outerStory using the
corresponding \FPT\ algorithm for \planarStory of
Binucci et al.~\cite[Theorem 4]{bdllmms-csp-JCSS24}.
A subset~$C$ of the vertex set of a graph~$G$ is a
\emph{vertex cover} of $G$ if every edge of~$G$ is incident to a
vertex in~$C$, and the \emph{vertex cover number} of $G$ is the
minimum size of a vertex cover of~$G$.
We use the same notation as Binucci et al.

Without loss of generality, we assume that the input graph $G$ does
not contain isolated vertices, as such vertices do not affect the
existence of a storyplan (of any kind).  Let $C$ be a vertex cover of
size $k$ of $G$.  For $U \subseteq C$, a vertex
$v \in V \setminus C$ is of \emph{type}~$U$ if $N(v) = U$, where
$N(v)$ denotes the set of neighbors of $v$ in $G$. This defines an
equivalence relation on $V \setminus C$ and in particular
partitions $V \setminus C$ into at most
$\sum_{i=1}^{k} {k \choose{i}}=2^{k}-1 < 2^k$ distinct types.
Let $V_U$ denote the set of vertices of type~$U$.

In order to obtain a kernel of $G$, Binucci et al.~\cite{bdllmms-csp-JCSS24} used the property of planar storyplans
for bipartite graphs (see \cref{le:bipartite-3})
to define three reduction rules.

\begin{enumerate}[(R1)]
\item {\em If there exists a type $U$ such that $|U| = 1$, then pick
    an arbitrary vertex $x \in V_U$ and remove it from
    $G$.}\label{rule:1}

\item {\em If there exists a type $U$ such that $|U| = 2$ and
    $|V_U| > 1$, then pick an arbitrary vertex $x \in V_U$ and remove
    it from $G$.}\label{rule:2}

\item {\em If there exists a type $U$ such that $|U| \ge 3$ and
    $|V_U| > 3$, then pick an arbitrary vertex $x \in V_U$ and remove
    it from $G$.}\label{rule:3}
\end{enumerate}

We modify the rules according to analogous properties of
outerplanar and forest storyplans
(see \cref{le:bipartite-small}).
To this end, we keep the rule \textsf{R\ref{rule:1}} for both
\forestStory and \outerStory, and we combine
rules~\textsf{R\ref{rule:2}} and~\textsf{R\ref{rule:3}} in a slightly
different manner for \forestStory (rule \textsf{R2$^\text{f}$}) and
for \outerStory (rule~\textsf{R2$^\text{o}$}).
We will see that, for the more restricted structures of forests and outerplanar graphs, it suffices
to have one rule (\textsf{R2$^\text{f}$} or \textsf{R2$^\text{o}$})
instead of two rules (\textsf{R\ref{rule:2}} and \textsf{R\ref{rule:3}}).

\begin{description}
\item[(R2$^\text{f}$)] {\em If there exists a type $U$ such that $|U| \ge 2$ and $|V_U| > 2$, then pick an arbitrary vertex $x \in V_U$ and remove it from $G$.}

\item[(R2$^\text{o}$)] {\em If there exists a type $U$ such that $|U| \ge 2$ and $|V_U| > 3$, then pick an arbitrary vertex $x \in V_U$ and remove it from $G$.}
\end{description}

Analogously to \cite[Lemma~7]{bdllmms-csp-JCSS24},
we next show that
a graph obtained from $G$ by applying one of the reduction
rules~\textsf{R\ref{rule:1}} and~\textsf{R2$^\text{f}$} (\textsf{R2$^\text{o}$}) admits a
forest (outerplanar) storyplan if and only if~$G$ does.
\begin{lemma}
\label{le:vc-reduction-rules}
Let $G'$ be a graph obtained from $G$ by applying one of the reduction
rules~\textsf{R\ref{rule:1}} and~\textsf{R2$^\text{f}$} (\textsf{R2$^\text{o}$}).
Then, $G'$ admits a forest (outerplanar) storyplan if and only if~$G$ does.
\end{lemma}

\begin{proof}
Clearly, $G'$ is a subgraph of~$G$.
Hence, if $G$ admits a forest (outerplanar) storyplan, then also $G'$ does due to \cref{prop:partial-graph}.
It remains to show that if $G'$ admits a forest (outerplanar) storyplan, so does $G$.

Let $x$ be the only vertex of $V(G) \setminus V(G')$.
If \textsf{R\ref{rule:1}} has been applied to obtain $G'$ from~$G$,
then the argument is the same as in \cite[Lemma 7]{bdllmms-csp-JCSS24}:
Consider the only neighbor~$u$ of~$x$ in~$G$.
In the storyplan $\mathcal{S}'$ of $G'$,
$u$ appears in some step $i_u$.
Obtain a forest (outerplanar) storyplan $\mathcal{S}$ of~$G$ from $\mathcal{S}'$ as follows.
Let $x$ appear in step $i_u + 1$
and let it disappear right after.
Draw $x$ sufficiently close to~$u$ such that the straight-line segment $ux$ does not interfere with the rest of the drawing.
The rest of the storyplan remains unchanged (just the indices of the time steps after $i_u$ increase by~1).

Next, consider the case that \textsf{R2$^\text{f}$} has been applied.
Let $U = \{u_1, u_2, \dots \} \subseteq C$ be the neighbors of $x$,
and let $y, z \in V(G') \setminus C$
be two ``twins'' of $x$
being of the same type $V_U$ as~$x$.
By \cref{le:bipartite-small},
either the vertices in $U$, or $y$ and~$z$ are simultaneously visible
in a forest storyplan $\mathcal{S}'$ of $G'$.
Suppose the vertices in~$U$ are all visible in the steps $[i, j]$.
Let $y$ and $z$ be visible in the steps
$[i_y, j_y]$ and $[i_z, j_z]$, respectively.
W.l.o.g., $i_y \le j_y < i_z \le j_z$.
Obtain a forest storyplan $\mathcal{S}$ of $G$ from $\mathcal{S}'$ as follows.
Let $x$ appear in step $j_y + 1$
and draw $x$ at the position of $y$
and let all edges to the neighbors in~$U$ coincide with those of~$y$.
Then, $x$ disappears in the next step.
The rest of the storyplan remains the same (but the indices increase by~1).
If $\mathcal{S}'$ was a forest storyplan, so is $\mathcal{S}$.
Now suppose that $y, z$ are both visible in the steps $[i, j]$.
Clearly, at most one vertex of~$U$ is visible at a time (see \cref{le:bipartite-small}).
For each step $s \le i$,
the new storyplan $\mathcal{S}$ is the same as $\mathcal{S}'$.
For each step $s > i$,
the new storyplan $\mathcal{S}$ is the same as $\mathcal{S}'$ in step $s - 1$
but with $x$ being additionally
visible in the steps $[i+1, j+1]$.
In these steps, $x$ is located next to~$y$.
In each step, there is at most one edge from $y$ to a vertex $u \in U$.
Add another edge from $x$ to $u$ that follows the curve of the edge $yu$
sufficiently close such that no intersections arise.
This shows the correctness of the rule \textsf{R2$^\text{f}$}.

We apply the analogous arguments to the rule \textsf{R2$^\text{o}$}.
The only difference is that we have besides~$y$ and~$z$ a third vertex~$z'$.
If the vertices in~$U$ are all visible simultaneously, then again draw $x$ at the position of~$y$.
If $y$, $z$, and $z'$ are visible simultaneously, then again draw~$x$ (and its incident edges) next to the location of~$y$ (and its incident edges).
If $\mathcal{S}'$ was an outerplanar storyplan, so is the resulting storyplan $\mathcal{S}$.
\end{proof}

Using \cref{le:vc-reduction-rules}, we now prove the analogous result to \cite[Theorem 4]{bdllmms-csp-JCSS24}:

\begin{theorem}
\label{thm:fpt-vc}
Deciding whether a graph~$G$ admits
a forest (or outerplanar) storyplan, and computing such a storyplan
(if any) can be done in $2^{2^{\mathcal{O}(k)}}+\mathcal{O}(n^2)$
time, where $n$ is the number of vertices of~$G$ and $k$ is the vertex
cover number of~$G$.
\end{theorem}

\begin{proof}
Compute a vertex cover~$C$ of size~$k$ in~$G$ in $\mathcal{O}(2^k + kn)$ time~\cite{DBLP:journals/tcs/ChenKX10}.
Then, group the vertices in $V(G) \setminus C$ by type into at most $n$ sets.
Apply the rules \textsf{R1} and \textsf{R2$^\text{f}$} (\textsf{R2$^\text{o}$}) exhaustively.
This can be one in $\mathcal{O}(n)$ time per application.
Each application reduces the number of vertices by~1,
so there are $\mathcal{O}(n)$ applications and a total running time of~$\mathcal{O}(n^2)$ for this step.
We obtain a kernel~$G^\star$ of size $\mathcal{O}(2^k)$.
In $G^\star$, compute a forest (outerplanar) storyplan by brute force:
Guess a total order of the vertices;
this defines the lifespan of every vertex and a subgraph for every frame.
For each step $i = 2, 3 \dots, |V(G^\star)|$,
maintain the set of all forest (outerplanar) embeddings
by taking the ones from step~$i-1$
and adding the vertex~$v$ that appears in step~$i$ at all possible positions in the rotation systems of $v$'s neighbors.
Discard the ones not being forest (outerplanar) embeddings and remove duplicates.
There are $|V(G^\star)|! \in (2^k)^{\mathcal{O}(k 2^k)} \subseteq 2^{2^{\mathcal{O}(k)}}$ total orders of the vertices.
The number of totals orders of the vertices also bounds the number of rotation systems.
Hence, this step takes $2^{2^{\mathcal{O}(k)}} \cdot 2^{2^{\mathcal{O}(k)}} = 2^{2^{\mathcal{O}(k)}}$ time in total.
If there is at least one forest (outerplanar) storyplan after step $|V(G^\star)|$,
there is also such a storyplan for~$G$
due to \cref{le:vc-reduction-rules}.
Note that the proof of \cref{le:vc-reduction-rules}
is constructive and a storyplan of~$G$
can be obtained from that of~$G^\star$ in $O(n^2)$ time.
\end{proof}

\subsection{\FPT\ Algorithm with Respect to the Feedback Edge Set Number}

Binucci et al.~\cite[Theorem 5]{bdllmms-csp-JCSS24} also derived
an \FPT\ algorithm with respect to the feedback edge set number.
A subset~$F$ of the edge set of a graph~$G$ is a
\emph{feedback edge set} if $G-F$ is acyclic,
and the \emph{feedback edge set number} of $G$ is the minimum
size of a feedback edge set of~$G$.  
In their algorithm, Binucci et al.\ defined the following reduction rules
to obtain a kernel of a graph $G$ of size $\mathcal{O}(f)$, where $f$
is the feedback edge set number of~$G$.

\begin{enumerate}[(RA)]
\item {\em If there exists a vertex
	of degree one, then remove it from $G$.} \label{rule:a}

\item {\em If, for $k \ge 3$, there
	exists a path with $k+2$ vertices such that its $k$
	inner vertices all have degree two, then remove its inner
	vertices from~$G$.} \label{rule:b}
\end{enumerate}

Their algorithm also works for \forestStory and
\outerStory:
After exhaustively applying both rules,
the resulting graph has size $O(f)$.
This means that we can guess a forest or outerplanar storyplan if it exists.
Into this storyplan, we re-insert the vertices that have been
removed by applying the rules as follows.
The next lemma corresponds to \cite[Lemma 8]{bdllmms-csp-JCSS24}.

\begin{lemma}
    \label{le:fes-reduction-rules}
    Let $G'$ be a graph obtained from $G$ by applying one of the reduction
    rules~\textsf{R\ref{rule:a}} and~\textsf{R\ref{rule:b}}.
    Then, $G'$ admits a forest (or outerplanar) storyplan if and only if~$G$ does.
\end{lemma}

\begin{proof}
Clearly, $G'$ is a subgraph of~$G$.
Hence, if $G$ admits a forest (outerplanar) storyplan, then also $G'$ does due to \cref{prop:partial-graph}.
It remains to show that if $G'$ admits a forest (outerplanar) storyplan, so does $G$.

If \textsf{R\ref{rule:a}} has been applied to obtain $G'$ from~$G$,
then the argument is the same as in \cref{le:vc-reduction-rules}:
Let $x$ be the only vertex of $V(G) \setminus V(G')$ and let $u$ be its only neighbor.
Obtain a forest (outerplanar) storyplan of~$G$ from such a storyplan of~$G'$ as follows.
Let $x$ appear right after~$u$ appears and let $x$ disappear in the next step.
Draw $x$ sufficiently close to~$u$ such that the straight-line segment $ux$ does not interfere with the rest of the drawing.
The rest of the storyplan remains unchanged.

Now suppose that \textsf{R\ref{rule:b}} has been applied to obtain $G'$ from~$G$
and let $\langle r, u_1, \dots, u_k, t \rangle$
be the path in~$G$ such that $r, t \in V(G')$ and $u_1, \dots, u_k \in V(G) \setminus V(G')$.
Let $\mathcal{S}'$ be a forest (or outerplanar) storyplan of~$G'$.
We next describe how to obtain a forest (or outerplanar) storyplan $\mathcal{S}$ of~$G$.
Let $[i_r, j_r]$ and $[i_t, j_t]$ be the lifespans of $r$ and $t$, respectively, in $\mathcal{S}'$.
W.l.o.g., $i_r < i_t$.
Let $u_1$ appear in step $i_r + 1$
and let $u_k$ appear in step $i_t + 2$.
Draw $u_1$ close to~$r$ and $u_k$ close to $t$
such that the straight-line segments $ru_1$ and $u_kt$ do not interfere with the rest of the drawing.
For each step $s \le i_r$, the new storyplan $\mathcal{S}$ is the same as $\mathcal{S}'$.
For each step $s \in [i_r + 1, i_t + 1]$,
$\mathcal{S}$ is the same as $\mathcal{S}'$ in step $s - 1$
with $u_1$ being additionally visible.
For each step $s > i_t + 1$,
$\mathcal{S}$ is the same as $\mathcal{S}'$ in step $s - 2$
with $u_1$ and $u_k$ being additionally visible.
Finally, we adjust the end of $\mathcal{S}$ by letting the remaining vertices appear in order:
for each $i \in [2, k-1]$,
let $u_i$ appear in step $n' + 1 + i$ where $n' = |V(G')|$.
Note that this is still a forest (or outerplanar) storyplan since after step $n' + 2$,
only $u_1$ and $u_k$ are still visible together with the vertices of the path $\langle u_2, \dots, u_{k-1} \rangle$.
It is easy to arrange $u_2, \dots, u_k$ along the line segment between $u_1$ and $u_k$
in a non-crossing way.
\end{proof}

This yields a linear-size kernel and leads to the following theorem,
which directly corresponds to a result by
Binucci et al.~\cite[Lemma 9 and Theorem 5]{bdllmms-csp-JCSS24}.

\begin{theorem}
\label{thm:fpt-fes}
Deciding whether a graph~$G$ admits
a forest (or outerplanar) storyplan, and computing such a storyplan
(if any) can be done in $2^{\mathcal{O}(f \log f)}+\mathcal{O}(n^2)$
time, where $n$ is the number of vertices of~$G$ and $f$ is the feedback edge set number of~$G$.
\end{theorem}

\begin{proof}
Compute a minimum-size feedback edge set~$F$ of~$G$ in linear time.
Note that this is possible by taking the complement of a (largest) spanning forest,
which is a linear-time problem.
Apply the rules \textsf{R\ref{rule:a}} and \textsf{R\ref{rule:b}} exhaustively.
This can be one in $\mathcal{O}(n)$ time per application.
Each application reduces the number of vertices by at least~1,
so there are $\mathcal{O}(n)$ applications and a total running time of~$\mathcal{O}(n^2)$ for this step.
Let the resulting graph be $G^\star$.
Let $F^\star \subseteq F$ be the set of edges of~$F$ in~$G^\star$.
Clearly, removing $F^\star$ from $G^\star$ yields a forest $H^\star$.
Observe that $H^\star$ has at most $2f$ leaves
because any leaf in~$H^\star$ is an endpoint of an edge in~$F^\star$
as otherwise it would not be in~$G^\star$ due to~\textsf{R\ref{rule:a}}.
Moreover, in any forest, the number of vertices with degree
at least~3 is at most the number of leaves minus two.
Hence, the number of degree-1 and degree-at-least-3 vertices in~$H^\star$
is at most $4f-2$.
Between those vertices, there are at most two neighboring vertices
of degree~2 due to \textsf{R\ref{rule:b}},
that is, at most $2(4f-3)$ degree-2 vertices.
Hence, $|E(H^\star)| < |V(H^\star)| \le 12f-8$ and $|E(G^\star)| < 13f-8$.
So, $G^\star$ is a kernel of size $\mathcal{O}(f)$.
As in \cref{thm:fpt-vc}, compute a forest (outerplanar) storyplan of~$G^\star$ by brute force:
Guess a total order of the vertices;
this defines the lifespan of every vertex and a subgraph for every frame.
For each step $i = 2, 3 \dots, |V(G^\star)|$,
maintain the set of all forest (outerplanar) embeddings
by taking the ones from step~$i-1$
and adding the vertex~$v$ that appears in step~$i$ at all possible positions in the rotation systems of $v$'s neighbors.
Discard the ones not being forest (outerplanar) embeddings and remove duplicates.
There are $|V(G^\star)|! \in 2^{\mathcal{O}(f \log f)}$ total orders of the vertices.
The number of totals orders of the vertices also bounds the number of rotation systems.
Hence, this step takes $2^{\mathcal{O}(f \log f)} \cdot 2^{\mathcal{O}(f \log f)} = 2^{\mathcal{O}(f \log f)}$ time in total.
If there is at least one forest (outerplanar) storyplan after step $|V(G^\star)|$,
there is also such a storyplan for~$G$
due to \cref{le:fes-reduction-rules}.
Note that the proof of \cref{le:fes-reduction-rules}
is constructive and a storyplan of~$G$
can be obtained from that of~$G^\star$ in $O(n^2)$ time.
\end{proof}

\section{Outerplanar Storyplans}
\label{sec:outer}

In this section, we present families of graphs that admit
outerplanar storyplans.

\begin{theorem}
	Every partial 2-tree admits a geometric outerplanar
	storyplan, and such a storyplan can be computed in linear time.
\end{theorem}

\begin{proof}
	Due to \cref{prop:partial-graph}, it suffices to prove the
	statement for 2-trees.
	
	Let $G$ be a 2-tree. Hence, there exists a \emph{stacking order}
	$\sigma = \langle v_1, \dots, v_n \rangle$ of the vertex set
	of~$G$. In other words, $G$ can be constructed as follows: we start
	with $v_1, v_2, v_3$ forming a $K_3$ and then, for $i \ge 4$, $v_i$
	is \emph{stacked} on an edge $v_k v_\ell$ with $k, \ell < i$, that
	is, $v_i$ is connected to $v_k$ and $v_\ell$ by edges.  We claim
	that we can choose a vertex order $\sigma'$ and an embedding
	$\mathcal{E}$ of $G$ such that $\sigma'$ (together with
	$\mathcal{E}$) defines an outerplanar storyplan.
	Moreover, we can obtain a straight-line drawing of $G$ with
	embedding $\mathcal{E}$ in linear
	time~\cite{fpp-hdpgg-Combin90,s-epgg-SODA90}.  Let~$\Gamma$ be such
	a drawing.  For the outerplanar storyplan that we construct we use
	the positions of vertices and edges as in $\Gamma$.  This yields a
	geometric storyplan.
	\cref{fig:2tree}(a) shows a 2-tree with a stacking order (that
	is not an outerplanar storyplan).
	
	\begin{figure}[bh]
		\centering \includegraphics{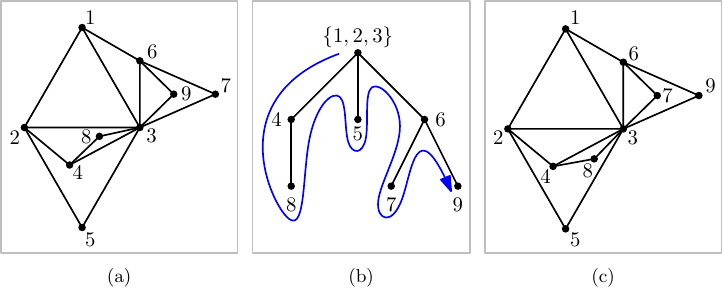}
		\caption{(a) Straight-line drawing of a 2-tree $G$ with stacking order $\sigma = \langle 1,2,3,4,8,5,6,7,9\rangle$; (b)~the stacking tree~$T_{G,\sigma}$
		of~$G$ and~$\sigma$ yields a vertex order~$\sigma'$;
		(c)~an embedding of~$G$ that together with~$\sigma'$ defines an
		outerplanar storyplan.}
		\label{fig:2tree}
	\end{figure}
    
	To show that an outerplanar storyplan always exists, consider the \emph{stacking
	tree}~$T_{G,\sigma}$ of~$G$.  The root of~$T_{G,\sigma}$
	represents the triangle $\Delta v_1 v_2 v_3$ given by the first
	three vertices of~$\sigma$.  For $i \in \{4, 5, \dots, n\}$,
        let $v_k v_\ell$ with
	$k < \ell < i$ be the edge on which~$v_i$ was stacked.  Then the node of $T_{G,\sigma}$ that represents~$v_i$ is a child of the node that represents~$v_\ell$.
    Note that if $\ell \le 3$, then~$v_i$ is a child of the root.
	\cref{fig:2tree}(b) shows the stacking tree of the 2-tree in	\cref{fig:2tree}(a).
	
	From $T_{G,\sigma}$, we obtain a vertex order
	$\sigma' = \langle v'_1, v'_2, \dots, v'_n \rangle$ that is an
	outerplanar storyplan as follows; see \cref{fig:2tree}(c).  Let
	$v_1'=v_1$, $v_2'=v_2$, and $v_3'=v_3$.  Now, we traverse $T_{G,\sigma}$ in depth-first order and, for each node~$\tau$ of $T_{G,\sigma}$, we add to~$\sigma'$ the vertex
	of~$G$ that $\tau$ represents.  We claim that for $\sigma'$, we can
	choose an embedding $\mathcal{E}$ (defined implicitly next) of $G$
	such that every frame is outerplanar.  Note that the first three
	vertices form a triangle, which always admits an outerplanar
	drawing.  Now consider $v_i'$ for $i \in \{4, 5, \dots, n\}$.  Our invariant
	is that, before the $i$-th frame starts, the parent~$p$ of~$v_i'$ in~$T_{G,\sigma}$ has degree~2 in the current outerplanar drawing and
	lies on the outer face.  This implies that $v_i'$ can be added to
	the outer face because it is stacked onto an edge of the outer face
	resulting again in an outerplanar drawing.  Of course, for $i = 4$,
	our invariant is satisfied.  If $p = v_{i-1}'$, then our invariant
	is trivially satisfied.  Otherwise, let $p = v_j'$ for some
	$j < i - 1$.  Observe that, for $k \in \{j + 1, \dots, i-1\}$, each
	$v_k'$ will have disappeared by the end of the $(i-1)$-th
	frame. This is due to the fact that $v_k'$ is not an ancestor of
	$v_i$, which means that all of the neighbors of $v_k'$ have already
	been introduced to the storyplan due to the depth-first pre-order
	traversal.  Essentially, every frame given by $\sigma'$ shows a
	subpath of $T_{G,\sigma}$, which is a sequence of stacked triangles
	admitting an outerplanar drawing.
\end{proof}

\begin{theorem}
	\label{thm:cubic-outerplanar}
	Every subcubic graph except $K_4$ admits a geometric
	outerplanar storyplan with at most five edges in each frame,
	and such a storyplan can be computed in linear time.
\end{theorem}

\begin{proof}
	Due to \cref{prop:partial-graph}, it suffices to prove the
	statement for cubic graphs.
	
	We can assume that the given cubic graph $G$ (which is not
        $K_4$) is connected.  Otherwise, we consider each connected
        component separately.  To construct an outerplanar storyplan,
        below, we will order the vertices $v_1, \dots, v_n$ of $G$
        such that the resulting sequence of graphs
        $\langle G_1, G'_1 \dots, G_{n-1}, G'_{n-1}, G_n \rangle$
	has the following property: for~$4 \le i \le n-1$, $G'_i$
        contains at most two edges.  Only~$G'_3$ may be a
        triangle and would thus contain three edges.        
	Then we show how to obtain outerplanar drawings
	$D_1, D'_1 \dots, D_{n-1}, D'_{n-1}, D_n$ of the graphs
	$G_1, G'_1 \dots, G_{n-1}, G'_{n-1}, G_n$, respectively.
	For $i \in [n]$, let $H_i = G[\{v_1, \dots, v_i\}]$.
	
	We pick the first vertex $v_1$ arbitrarily.  For $2 \le i \le
        n$, consider the vertices of~$G'_{i-1}$ that have maximum
        degree in~$H_{i-1}$.  Among those, let $v$ be a vertex that
        additionally has the largest degree in~$G'_{i-1}$.
	Then, let $v_i$ be a neighbor of~$v$ in~$G-\{v_1,\dots,v_{i-1}\}$.
        Note that $v$ always has such a neighbor, otherwise $v$
	would already be completed and, hence, would not be in $G'_{i-1}$.
	The intuition behind this choice is that we want to remove~$v$
	from the drawing as soon as possible.

        Note that, for every $i \in [n]$, the graph~$H_i$ is
        connected.  This is due to the fact that $G$ is connected and
        the new vertex~$v_i$ is always adjacent to a vertex
        in~$G'_{i-1}$.
	
	We claim that, for $4 \le i \le n-1$, the graph $G'_i$ contains at
	most two edges and, if $G'_i$ contains exactly two edges, then
	both edges are incident to~$v_i$.  If our claim holds, then, for
	$i \in [n]$, $G_i$ contains at most five edges.  Indeed, even if
	$G'_3$ has three edges (that is, $G'_3$ is a triangle; see
	\cref{storyplan:fig:cubic-outerplanar-special}), then $G_4$
        still has at most five edges since~$G$ is not~$K_4$.

        We show our claim by induction over~$i$.

        For the base of the induction, observe that $G_1'$ ($=G_1$)
        contains only vertex~$v_1$ and no edges and $G_2'$ ($=G_2$) is
        the edge~$v_1v_2$.  If $G_3'$ ($=G_3$) is a triangle, then
        $G_4'$ has at most one edge (see
        \cref{storyplan:fig:cubic-outerplanar-special}).  Otherwise
        $G_3'$ is the path $\langle v_1,v_2,v_3 \rangle$, $G_4$ is a
        3-star with apex~$v_2$, and $G_4'$ consists of the three
        independent leaves of~$G_4$.
	
	\begin{figure}[tb]
          \centering
          \includegraphics{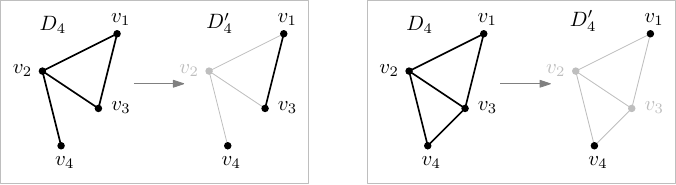}
          \caption{Illustration of the special case that $G'_3$
            contains three edges.  Then $v_4$ has one neighbor
            in~$H_3$ (left subfigure) or two neighbors in~$H_3$ (right
            subfigure), but $G'_4$ has at most one edge.  Gray
            vertices and edges are not part of the current graphs, but
            were visible in the previous step.}
          \label{storyplan:fig:cubic-outerplanar-special}
	\end{figure}
	
        Now, for the inductive step, consider $i \ge 4$.
	There are three cases depending on the degree of~$v$ 
	in~$G'_{i-1}$; see \cref{fig:cubic-outerplanar}.
	
	\begin{enumerate}[(C1)]%
		\item Vertex~$v$ does not have any neighbors in~$G'_{i-1}$.

		By the choice of $v$, this implies that there are no edges
		in~$G'_{i-1}$.
		Hence, all edges in~$G_i$ are new and incident to~$v_i$.
                If $v_i$ has three neighbors in $G_i$, then $v_i$ will
		disappear, and there are no more edges in $G'_i$.
                Otherwise, $G'_i$ has at most two edges, which are
                both incident to~$v_i$.
		
		\item Vertex~$v$ has one neighbor in $G'_{i-1}$.

		If $v$ has degree~2 in $H_{i-1}$, then $v$ disappears in the next
		step and $G'_i$ does not contain it.  Since $v_i$ has at most one
		edge that stays in $G'_i$, the number of edges in $G'_i$ is not
		larger than in $G'_{i-1}$.  If $v$ has degree 1 in $H_{i-1}$, then,
		by the choice of~$v$, all other vertices
                in~$G_{i-1}$ have degree at
		most~1 in $H_{i-1}$.  Hence, $i=3$, that is, $v$ and its neighbor
		are the first two vertices that we introduced.
		
		\item Vertex~$v$ has two neighbors in $G'_{i-1}$.

		In this case, the two edges incident to~$v$ are the only edges
		in~$G'_{i-1}$. Then $v$
		disappears as $v_i$ is its last neighbor. Therefore, $G'_i$ contains at most one edge that $v_i$
		may have introduced.
	\end{enumerate}
	
	\begin{figure}[tb]
		\centering
		\captionsetup[subfigure]{justification=centering}
		\begin{subfigure}[t]{.3\linewidth}
			\centering
			\includegraphics[page=1]{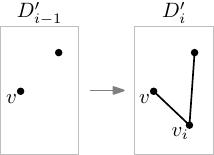}
			\caption{case (C1)}
			\label{fig:cubic-outerplanar1}
		\end{subfigure}
		\hfill
		\begin{subfigure}[t]{.3\linewidth}
			\centering \includegraphics[page=2]{cubic-outerplanar.pdf}
			\caption{case (C2)}
			\label{fig:cubic-outerplanar2}
		\end{subfigure}
		\hfill
		\begin{subfigure}[t]{.3\linewidth}
			\centering \includegraphics[page=3]{cubic-outerplanar.pdf}
			\caption{case (C3)}
			\label{fig:cubic-outerplanar3}
		\end{subfigure}
		\caption{Cases considered in the proof of
                  \cref{thm:cubic-outerplanar}.  In each of them, we
                  depict a situation where the number of edges
                  in~$G'_i$ is maximized.  Gray vertices and edges
                  were visible in some previous steps.}
		\label{fig:cubic-outerplanar}
	\end{figure}
	
	We have shown that, in each case, the number of visible edges in
	$G'_i$, for $4 \le i \le n-1$, is at most two.  Note that, if there
	are two edges, then they share an endpoint.
	
	It is obvious that the drawings $D_1$, $D'_1$, $D_2$, $D'_2$,
        $D_3$, and~$D'_3$ are outerplanar even if they are
        straight-line.  Now we want to show, for
        $i \in \{4,\dots,n\}$, how to place $v_i$ such that $D_{i}$ is
        outerplanar and straight-line.  We assume that $v_i$ is
        connected to three visible vertices $u, v, w \in V(G'_{i-1})$.
        The other cases are easier and are covered by this case.
	
	We know that there are at most two edges in $G'_{i-1}$. 
	If there are edges, we may assume that they are connected to~$v$.  
	We consider the case where there are exactly two edges~$(v,v')$ and~$(v,v'')$ in $G'_{i-1}$; the other cases are covered by this one.
	Note that possibly $u$ is $v'$ and/or $w$ is~$v''$.  
	In this case, the triangle $\triangle uvv_i$ and/or the triangle $\triangle vwv_i$ would appear in~$D_i$.  
	In any case, we place $v_i$ in the vicinity of~$v$ such that none of the new
	edges intersects the visible ones and no visible vertex lies in one of the
	triangles $\triangle uvv_i$ and $\triangle vwv_i$ that we have potentially created.  
	It is easy to see that such a placement exists.
	
	Now assume that $u$, $w$, $v'$, and $v''$ are pairwise different vertices.  
	In this case $G_i$ is a tree, and we need to avoid only edge intersections.
	Let $\ell_{v'}$ and $\ell_{v''}$ be rays that start in~$v$ and
        go through~$v'$ and~$v''$, respectively.
        There are three cases: $\triangle vuw$
	intersects neither $\ell_{v'}$ nor $\ell_{v''}$, it intersects
	both, or it intersects exactly one of them.
	
	\begin{figure}[t]
          \centering
          \begin{subfigure}[t]{.4\linewidth}
            \centering
            \includegraphics[page=1]{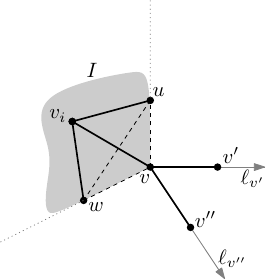}
            \caption{$\triangle vuw$ intersects neither $\ell_{v'}$
              nor $\ell_{v''}$.}
            \label{fig:cubic-straight-line-proof1}
          \end{subfigure}
          \hfil
          \begin{subfigure}[t]{.4\linewidth}
            \centering
            \includegraphics[page=2]{cubic-straight-line-proof}
            \caption{$\triangle vuw$ intersects both $\ell_{v'}$ and
              $\ell_{v''}$.}
            \label{fig:cubic-straight-line-proof2}
          \end{subfigure}
                      
          \medskip
                
          \begin{subfigure}[t]{.4\linewidth}
            \centering
            \includegraphics[page=4]{cubic-straight-line-proof}
            \caption{$\triangle vuw$ intersects only $\ell_{v'}$, $uw$
              does not intersect $(v,v')$.}
            \label{fig:cubic-straight-line-proof3}
          \end{subfigure}
          \hfil
          \begin{subfigure}[t]{.4\linewidth}
            \centering \includegraphics[page=3]{cubic-straight-line-proof}
            \caption{$\triangle vuw$ intersects only $\ell_{v'}$,
              $(u,w)$ intersects~$(v,v')$.}
            \label{fig:cubic-straight-line-proof4}
          \end{subfigure}
          \caption{Various cases for the placement of~$v_i$.  In each
            case, $v_i$ can be placed in the gray region~$I$ to
            obtain a straight-line outerplanar
            drawing~$D_i$.}
          \label{fig:cubic-straight-line-proof}
	\end{figure}
	
	In the first two cases, we place $v_i$ in the intersection~$I$ of the
	open halfplanes bounded by the lines through $(v,u)$ and $(v,w)$
	that contain neither $v'$ nor~$v''$; see the gray regions in
	\cref{fig:cubic-straight-line-proof1,fig:cubic-straight-line-proof2}.
        Now assume that $\triangle vuw$
	intersects one of the rays, say,~$\ell_{v'}$.  
	Let~$\ell_u$ be the line that goes through~$u$ and~$v'$,
        and let~$\ell_w$ be the line that goes through~$w$ and~$v'$.  
	Let~$h_u$ and $h_w$ be the halfplanes bounded by~$\ell_u$
        and~$\ell_w$, respectively, that do not contain~$v$.  
	Let~$h_{v''}$ be the halfplane bounded by $\ell_{v''}$
        that contains~$v'$.
        Let $I = h_u \cap h_w \cap h_{v''}$.
        Note that $h_u \cap h_w$ is a non-empty wedge with apex~$v'$.
        Since $\ell_{v''}$ does not go through~$v'$ but $h_{v''}$
        contains~$v'$, $h$ is either a wedge or a (possibly unbounded)
        polygon with three sides.  Hence, the interior~$I^-$ of
        $I \setminus \ell_{v'}$ is non-empty, and we place~$v_i$ in~$I^-$;
        see \cref{fig:cubic-straight-line-proof3}, where $(u,w)$ does not
        intersect $(v,v')$, and \cref{fig:cubic-straight-line-proof4},
        where $(u,w)$ intersects $(v,v')$.
        None of the new edges incident to~$v_i$ contains a vertex
        visible in~$D_i$.  Hence, we have found a position for~$v_i$
        such that the drawing~$D_i$ is straight-line and outerplanar.
	
	To argue the linear runtime, note that, for every~$i \in [n]$,
        $G_i$ has constant size.  Therefore, we can pick~$v_i$ in
        constant time.  Since the degree of~$v_i$ is at most~$3$, we
        can update~$H_i$ in constant time by using a standard graph
        data structure.%
\end{proof}

\section{Forest Storyplans}
\label{sec:forest}

Clearly, any triangle is an obstruction for a graph to admit a forest
storyplan.  Interestingly, for planar and subcubic graphs this is the
only obstruction for the existence of a forest storyplan as we show
now.

\begin{theorem}
	\label{thm:subcubic}
	Every $\triangle$-free subcubic graph admits a geometric forest
	storyplan.  Such a storyplan can be computed in linear time and
	has at most five edges per frame.
\end{theorem}

\begin{proof}
	Due to \cref{prop:partial-graph}, it suffices to prove the
	statement for $\triangle$-free cubic graphs.
	Recall the proof of \cref{thm:cubic-outerplanar}.
	In that proof, we showed that, for $4 \le i \le n-1$, there are at
	most two edges in~$G'_i$, and two edges may appear only if they
	share a vertex.  In such a case, we always pick a shared vertex as a
	neighbor of the next vertex.  Since our graph is $\triangle$-free,
	adding a vertex will never make a cycle visible.  Moreover,
	$G'_3$ cannot be a triangle and, thus, for every~$i \in [n-1]$,
	$G_i'$ contains at most two edges.
\end{proof}

As a warm-up for our main result regarding $\triangle$-free planar
graphs (\cref{thm:planar}), we first present an algorithm for
constructing storyplans for $\triangle$-free {\em outer}planar
graphs.

\begin{proposition}
  Every $\triangle$-free outerplanar graph admits a geometric
  forest storyplan, and such a storyplan can be computed in linear
  time.
\end{proposition}

\begin{proof}
  Let~$G$ be a $\triangle$-free outerplanar graph, and let~$\Gamma$ be
  an outerplanar straight-line drawing of $G$.
  Let~$\sigma = \langle v_1, v_2, \dots, v_n \rangle$ be the circular
  order of the vertices along the outer face of~$\Gamma$ (which can
  easily be determined in linear
  time~\cite{BattistaF05,BattistaFrati2006}).  We claim that $\sigma$
  yields a forest storyplan of~$G$.  (The resulting storyplan will
  also be geometric because every frame is a subdrawing of~$\Gamma$.)
	
	To this end, we show that there is no frame where a complete face of
	$\Gamma$ is visible.  If this is true, then no frame contains a
	complete cycle.  This is due to the fact that, in outerplanar
	graphs, the vertex set of every cycle contains the vertex set of at
	least one face.
	Let~$F = \langle v_{i_1}, v_{i_2}, \dots, v_{i_k} \rangle$ with
	$i_1 < i_2 < \dots < i_k$ be a face of~$G$.  Since $G$ is
	$\triangle$-free, we have $k \ge 4$.  Note that $v_{i_1}$ and $v_{i_3}$
	as well as $v_{i_2}$ and $v_{i_4}$ are not adjacent.  Since $G$ is
	outerplanar, $v_{i_2}$ may be adjacent only to vertices that appear
	in~$\sigma$ between (and including) $v_{i_1}$ and~$v_{i_3}$.  Therefore,
	$v_{i_2}$ disappears before $v_{i_4}$ appears.  Hence it is indeed
	not possible that all vertices of the same face appear in a frame.
\end{proof}

Now we generalize the simple result above from $\triangle$-free
outerplanar graphs to $\triangle$-free planar graphs.
Note, however, that we do not guarantee linear running time any more.

\begin{theorem}
	\label{thm:planar}
	Every $\triangle$-free planar graph admits a geometric forest
	storyplan, and such a storyplan can be computed in quadratic time.
\end{theorem}

\begin{proof}
	Let $G$ be a $\triangle$-free planar graph, and let $\Gamma$ be a
	planar straight-line drawing of $G$.
	In the desired forest storyplan for~$G$, we use the positions of the
	vertices in $\Gamma$.
	
	We first give a rough outline
	of our iterative algorithm and then describe the details.  In each
	iteration (which spans one or more steps of the storyplan that we construct),
	we \emph{pick} a vertex on the current outer face
    (we explain below what the current outer face is).
    We add the picked vertex and its neighbors (if they are not visible yet) to
	the storyplan one by one.  In this way, after each iteration,
	at least one vertex disappears, namely the one we picked.

    For the ease of notation, we slightly adjust the definition of~$G_i$ in this proof.
    Namely, we let $G_1=G$ and, for $i \in \{1,2,\dots\}$, we let $v_i$ be the vertex
	that we pick in iteration~$i$.
    Then, we define $G_{i+1}$ as the subgraph of~$G_{i}$
	that we obtain after removing the vertices (and the edges incident
	to them) that disappear in iteration~$i$; see
	\cref{fig:triangle-free-planar-iteration2}.  The algorithm
	terminates as soon as $G_i$ is a forest and adds the remaining
	vertices in arbitrary order to the storyplan under construction.
    When we speak about the (current) outer face in step~$i$,
    we mean the outer face of~$G_i$ in the drawing derived from~$\Gamma$.
	We call vertices and edges incident to the (current)
	outer face \emph{outer}. The others are \emph{inner}.
	
	We always pick outer vertices.  For this reason, only two
        types of vertices can potentially create cycles: the endpoints
        of \emph{chords} (i.e., inner edges incident to two outer
        vertices) and the endpoints of \emph{half-chords} (i.e.,
        length-2 paths that connect two outer vertices via an inner
        vertex).
	
	\begin{figure}[t]
		\centering
		\includegraphics{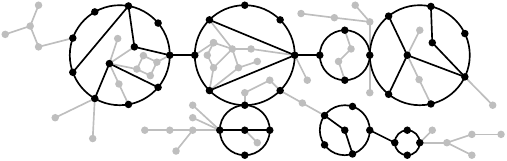}
		\caption{From an embedded $\triangle$-free planar graph $G_i$ (black \& gray),
			we obtain $G_i'$ (black).
			Note that $G_i'$ decomposes into seven simple cycles and two connected components.
			The outer edges and vertices of these connected components form cactus graphs.}
		\label{fig:gi-prime}
	\end{figure}
	
	Let~$G_i'$ be the (embedded) subgraph of~$G_i$ (embedded according
	to~$\Gamma$) that consists of all vertices and edges that lie on a
	simple cycle that bounds the outer face of~$G_i$,
	plus every edge that connects two cycles,
	plus all chords and half-chords (and, thus, plus the
	inner vertices that lie on the half-chords) of~$G_i$; see
	\cref{fig:gi-prime,fig:triangle-free-planar-dual}.  For example,
	the edges~$e$ and~$e'$ of~$G_2$ in
	\cref{fig:triangle-free-planar-iteration2} are not part of~$G_2'$.
	We say that a vertex of~$G_i'$ is \emph{free} if it lies on the
	outer face and is not part of a chord or a half-chord.
	
	Let $H_i$ be the weak dual of~$G_i'$ (see
	\cref{fig:triangle-free-planar-dual}), i.e., the (embedded)
	multigraph that has a vertex for each inner face of~$G_i'$ and an
	edge for each pair of inner faces that are incident to a common edge
	of~$G_i'$.  Note that $H_i$ is outerplane
    and even has all edges on the outer face
    (since the inner vertices of $G_i'$ form an independent set) and that $H_i$ has no loops (since $G_i'$ does not
	have leaves).
	\begin{figure}[t]
		\centering
		\begin{subfigure}[t]{.31\linewidth}
			\centering \includegraphics[page=1]{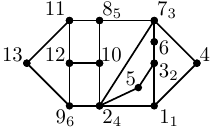}
			\caption{$\triangle$-free planar graph $G$ with a forest storyplan}
			\label{fig:triangle-free-planar-example}
		\end{subfigure}
		\hfill
		\begin{subfigure}[t]{.320\linewidth}
			\centering \includegraphics[page=2]{triangle-free-planar-algo}
			\caption{iteration~2: subgraph~$G_2$ (black) of $G_1 = G$}
			\label{fig:triangle-free-planar-iteration2}
		\end{subfigure}
		\hfill
		\begin{subfigure}[t]{.305\linewidth}
			\centering \includegraphics[page=3]{triangle-free-planar-algo}
			\caption{subgraph~$G_2'$ of~$G_2$ with weak dual~$H_2$
				(green)}
			\label{fig:triangle-free-planar-dual}
		\end{subfigure}
		\caption{A $\triangle$-free planar graph $G$
			where (a) shows a forest storyplan
			computed by our algorithm, (b) shows the
			result of the first iteration of the algorithm,
			and (c) shows the auxiliary graph for the second iteration.
			Subscripts refer to the iteration in which a vertex is picked. 
			Red crosses mark vertices that may not be picked.
		}
		\label{fig:triangle-free-planar-algo}
	\end{figure}
	We maintain the following invariants:
	\begin{enumerate}[({I}1)]%
		\item \label{inv:cycle}%
		At no point in time, the set of visible edges on the outer
		face forms a cycle.
		\item \label{inv:vertices}%
		{\em During} iteration~$i$, the only inner vertices that may be
		visible are those that are adjacent to~$v_i$ and to no other
		visible vertex on the outer face. 
		\item \label{inv:edges}%
		{\em During} iteration~$i$, the only inner edges that may be
		visible are those that are incident to~$v_i$ and to no other
		vertex on the outer face.
		\item \label{inv:visible}%
		{\em At the end} of each iteration (after removing the vertices
		that are not visible any more and before picking a new one), only
		vertices and edges incident to the outer face are visible.
	\end{enumerate}
	
	Obviously, if the invariants hold, the set of visible edges in each
	frame forms a forest.  In order to guarantee that the invariants
	hold, we use the following rules that determine which vertices we
	may not pick; see \cref{fig:rules-planar-proof}.
	We call a vertex observing these rules \emph{good}.
	Note that we always pick a good vertex
	on the outer face of $G_i'$~--
	we will later argue that there is always one.
	\begin{enumerate}[{Rule} 1:]%
		\item \label{rule:invisible}%
		Do not pick a vertex~$v$ whose closed
		neighborhood~$N[v]=\{v\} \cup \{ u \colon uv \in E(G_i) \}$
		contains all invisible vertices of the outer face of~$G_i'$.
		\item \label{rule:chord}%
		Do not pick an endpoint of a chord.
		\item \label{rule:neighbor}%
		Do not pick a neighbor of an endpoint of a chord if the other
		endpoint of that chord is visible.
		\item \label{rule:half-chord}%
		Do not pick an endpoint of a half-chord if the other endpoint
		is visible.
	\end{enumerate}
	
	\begin{figure}[t]
		\centering
		\begin{subfigure}[t]{.24\linewidth}
			\centering \includegraphics[page=1]{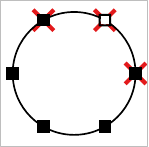}
			\caption{rule~\ref{rule:invisible}}
			\label{fig:rule1-planar-proof}
		\end{subfigure}
		\hfill
		\begin{subfigure}[t]{.24\linewidth}
			\centering \includegraphics[page=2]{rules-planar-proof}
			\caption{rule~\ref{rule:chord}}
			\label{fig:rule2-planar-proof}
		\end{subfigure}
		\hfill
		\begin{subfigure}[t]{.24\linewidth}
			\centering \includegraphics[page=4]{rules-planar-proof}
			\caption{rule~\ref{rule:neighbor}}
			\label{fig:rule3-planar-proof}
		\end{subfigure}
		\hfill
		\begin{subfigure}[t]{.24\linewidth}
			\centering \includegraphics[page=3]{rules-planar-proof}
			\caption{rule~\ref{rule:half-chord}}
			\label{fig:rule4-planar-proof}
		\end{subfigure}
		\caption{Rules that determine which vertices may not be picked
			(marked by red crosses). Black squares represent visible
			vertices, white squares represent invisible vertices, and gray
			disks represent vertices that may be visible or invisible.%
		}
		\label{fig:rules-planar-proof}
	\end{figure}

	Rule~\ref{rule:invisible} ensures that we do not close a cycle on
	the outer face, thus, invariant~\ref{inv:cycle} holds.
	Rule~\ref{rule:half-chord} ensures that none of the visible inner
	vertices is adjacent to two visible vertices on the outer face
	(including the picked vertex), thus, invariant~\ref{inv:vertices} holds.
	Rules~\ref{rule:chord} and~\ref{rule:neighbor} ensure that no chords
	are visible.  Together with rule~\ref{rule:half-chord} and the fact
	that $G$ is $\triangle$-free, they ensure that the inner edges
	that are visible are incident to the picked vertex
	and no other vertex on the outer face.
	Thus, invariant~\ref{inv:edges} holds.
	Invariant~\ref{inv:visible} holds because we always pick a vertex
	on the outer face and remove it.  As a result, the faces incident to
	the picked vertex become part of the outer face and the previously
	inner neighbors (if any) of the picked vertex become incident to the outer face.

	It remains to prove that, as long as $G_i$ is not a forest (and the
	algorithm terminates),
	there exists a vertex that can be picked without violating any of our rules.
	Our proof is constructive; we show how to find a vertex to pick.
	
	We first show that $H_i$ is a (collection of) cactus graph(s), that is,
	every edge of~$H_i$ lies on at most one cycle.
	Suppose that $H_i$ contains an
	edge~$e_1$ that lies on at least two simple cycles.  If the
	interiors of the two cycles are disjoint, then $e_1$ is not incident
	to the outer face of~$H_i$ (contradicting $H_i$ being outerplane).
	Otherwise, one of the cycles has at
	least one edge~$e_2 \ne e_1$ in the interior of the other cycle, again contradicting $H_i$ being outerplane.

	We show in two steps that $G_i$ (actually even $G_i'$) always contains a
	good vertex, which we pick.
	First, we show how to find a good vertex in the base case, that is, if
	the outer face of $G_i'$ is a simple cycle.
	Then, we consider the general case where the outer face of $G_i'$ is
	a (collection of) cactus graph(s).
	Here, we repeatedly apply the argument of the base case to find a good vertex.
	So, assume that the outer face of $G_i'$ is a simple cycle and, hence,
	$H_i$ is connected.
	
	In the trivial case that 
	the weak dual 
	$H_i$ is a single vertex,
	$G_i'$ is a cycle of at least four free vertices.
	Due to invariant \ref{inv:cycle}, there is an invisible vertex~$v \in G_i'$.
	Any non-neighbor of~$v$ in $G_i'$ is a good vertex, which we can pick.
	
	If $H_i$ has a vertex of degree~1,
	which corresponds to a face~$f$ of~$G_i'$,
	it means that $f$ is incident to exactly one chord and to no half-chords.
	Since $G$ is $\triangle$-free, there are at least two free vertices
	in~$f$.  Note that at most one endpoint of the chord is visible
	(due to invariant~\ref{inv:edges}).  If one endpoint is indeed
	visible, then its unique neighbor on the boundary of~$f$ that is not
	incident to the chord observes all rules and can be picked. 
	If none of the endpoints
	of the chord is visible, then any free vertex of~$f$ can be picked.
	
	Otherwise, all vertices of $H_i$ have degree at least~2.
	Let $F$ be the set of faces of $G_i'$
	that are incident to exactly one half-chord
	and to an arbitrary number of outer edges (but to no other inner edge).
	Note that in $H_i$, $F$ corresponds to a set of vertices of degree~2.
	We now use the following two helpful claims.
	
	\begin{claim}
		\label[claim]{clm:FAtLeastTwo}
		The set $F$ has cardinality at least $2$.
	\end{claim}
	
	\begin{claimproof}
		Consider the \emph{block-cut tree} of $H_i$, that is, the tree that
		has a node for each cut vertex of~$H_i$ and a node for each
		2-connected component (called \emph{block}) of~$H_i$.  A block node
		and a cut-vertex node are connected by an edge in the block-cut tree
		if, in~$H_i$, the block contains the cut vertex.  Clearly, every
		leaf of the block-cut tree is a block node, and every block of~$H_i$
		is either a cycle or an edge.  Consider any leaf of the block-cut
		tree.  It is a block node representing a cycle of~$H_i$ as otherwise
		it would represent an edge of~$H_i$ having an endpoint that is a
		vertex of degree~1 in $H_i$.  Clearly, there is an inner vertex
		of~$G_i'$ being incident to all faces represented by that cycle
		and all edges of that cycle are then dual to half-chord edges of
		$G_i'$.  As each such cycle contains at least two vertices but at
		most one of them is a cut vertex in~$H_i$, there is, per leaf of the
		block-cut tree, at least one vertex of degree two in $H_i$ whose
		incident edges are both dual to a half-chord of~$G_i'$.  The
		block-cut tree is either a single block node representing a cycle
		in~$H_i$, or the block-cut tree has at least two leaf nodes.  Thus,
		$|F| \ge 2$.
	\end{claimproof}
	
	\begin{claim}
		\label[claim]{clm:bothSidesF}
		Let the edge $e$ (or the edge pair $\{e_1, e_2\}$) be any chord (or half-chord) of~$G_i'$,
		let $F_1$ be the set of inner faces on the one side,
		and let $F_2$ be the set of inner faces on the other side of $e$ (or $\{e_1, e_2\}$, resp.).
		Then, $F_1 \cap F \ne \emptyset$ and $F_2 \cap F \ne \emptyset$.
	\end{claim}
	
	\begin{claimproof}
		If we have a chord $e$, then the dual edge of $e$ is an edge of
		$H_i$, which corresponds to a block node of the block-cut tree,
		which cannot be a leaf (because all leaves represent cycles).  Then,
		however, if we traverse the block-cut tree, in the one or the other
		direction of the edge, we will find a leaf in both parts and each
		leaf contains a vertex in~$F$ as shown in the proof
		of~\cref{clm:FAtLeastTwo}.
		
		If we have a half-chord $\{e_1, e_2\}$, then the dual edges divide a
		cycle~$C$ of~$H_i$ into two.  If one of the resulting parts of~$C$
		does not contain a cut vertex, then it contains only vertices from
		$F$ by definition.  Otherwise, the previous argument applies: the
		block-cut tree gets divided into two parts and each part needs to
		contain a leaf of the block-cut tree.
	\end{claimproof}
	
	We continue to show that there is a good vertex
	on the outer face of~$G_i'$, which we can pick.
	Assume first that $G_i'$ does not have chords.
	Thus, all vertices of $G_i'$ trivially observe rules~\ref{rule:chord} and~\ref{rule:neighbor}.
	Let $f \in F$, and let $u$ and $w$ be the endpoints of the unique half-chord incident to~$f$.
	If there is a free vertex~$v$ in~$f$ such that $N[v]$ does not contain the last invisible
	vertices of the outer face of~$G_i'$, then we pick~$v$. Rules~\ref{rule:invisible} and
	\ref{rule:half-chord} are observed by the definition of~$v$.
    If, for every face~$f \in F$ and for every free vertex~$v$ in $f$, $N[v]$ contains all invisible vertices of
	the outer face of~$G_i'$, consider the following three cases; see \cref{fig:planar-correctness}.
	The cases are ordered by priority; if we fulfill
	the conditions of multiple cases, the first case applies.
	
	\begin{figure}[t]
		\centering
		\begin{subfigure}[t]{.32\linewidth}
			\centering \includegraphics[page=1]{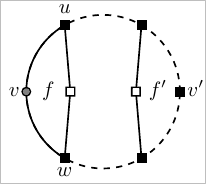}
			\caption{case~\ref{case:simple-both-visible}}
			\label{fig:planar-correctness-case1}
		\end{subfigure}
		\hfill
		\begin{subfigure}[t]{.32\linewidth}
			\centering \includegraphics[page=2]{planar-correctness}
			\caption{case~\ref{case:simple-one-visible}}
			\label{fig:planar-correctness-case2}
		\end{subfigure}
		\hfill
		\begin{subfigure}[t]{.32\linewidth}
			\centering \includegraphics[page=3]{planar-correctness}
			\caption{case~\ref{case:simple-none-visible}}
			\label{fig:planar-correctness-case3}
		\end{subfigure}
		\caption{Cases when there is no chord in~$G_i'$. We always find a good vertex.}
		\label{fig:planar-correctness}
	\end{figure}
	
	\begin{enumerate}[(C1)]%
		\item \label{case:simple-both-visible}
		Both $u$ and $w$ are visible.
		Then, consider a face $f' \in F$ different from~$f$,
		which exists by \cref{clm:FAtLeastTwo}.
		Clearly, all of its vertices are visible, and
		we can pick any free vertex~$v'$ of $f'$ 
		without violating the rules.
		\item \label{case:simple-one-visible}
		Exactly one of $\{u, w\}$ is visible.
		Without loss of generality,
                assume that $u$ is visible and $w$ is invisible.
		We claim that $u$ observes all rules.
		Since $w$ remains invisible after picking $u$,
		$u$ observes rule~\ref{rule:invisible}.
		If there was another half-chord incident to~$u$,
		either it would again be incident to $w$, which does
		not violate rule~\ref{rule:half-chord}, or it would be
		incident to another vertex of $G_i'$, which is visible.
		By \cref{clm:bothSidesF}, however,
		there is another face~$f' \in F$
		on the other side of that half-chord.
		As all of the vertices of~$f'$ on the outer face
		are visible, we would be in case~\ref{case:simple-both-visible} instead.
		\item \label{case:simple-none-visible}
		Both $u$ and $w$ are invisible.
		We claim that $u$ observes the rules.
		Similar to case~\ref{case:simple-one-visible},
		$u$ observes rule~\ref{rule:invisible} (since $w$ stays invisible)
		and rule~\ref{rule:half-chord}
		(if there was another half-chord incident to $u$
		whose other endpoint is visible, we would be in
		case~\ref{case:simple-both-visible} or in
		case~\ref{case:simple-one-visible}).
	\end{enumerate}
	
	Now assume that $G_i'$ has one or more chords.
	Of course, each of these chords has at most one visible endpoint.
	The chords with exactly one visible endpoint divide $G_i'$
	into several subgraphs.
	Observe that at least one of these subgraphs
	contains no such chord in its interior
	and is bounded by only one of them
	(or no chord has a visible endpoint, then there is only one subgraph, namely $G_i'$).
	We call this subgraph $\hat G_i'$
	and we let $u$ and $w$ denote the visible and invisible endpoints
	of the bounding chord~$e$, respectively
	(or if there is only one subgraph, then $u$ and $w$ are just neighbors).
	By a case distinction on the faces incident to $u$ and $w$,
	we can show that there is always a good vertex on the outer face of~$\hat G_i'$,
	and hence on the outer face of~$G_i'$.
	
	Let $\hat F$ be the subset of~$F$ that is contained in $\hat G_i'$,
	which is non-empty by \cref{clm:bothSidesF}.
	If there is a face $f \in \hat F$ neither incident to $u$ nor $w$,
	then we can pick a free vertex $v$ of $f$;
	see \cref{fig:planar-correctness-chord-case1}.
	Since at least $u$
	stays invisible, $v$ observes rule~\ref{rule:invisible}.
	Trivially, the other rules are also observed.
	
	\begin{figure}[t]
		\centering
		\captionsetup[subfigure]{justification=centering}
		\begin{subfigure}[t]{.37\linewidth}
			\centering \includegraphics[page=1]{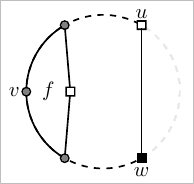}
			\caption{$f$ is not incident to $u$ or~$w$.}
			\label{fig:planar-correctness-chord-case1}
		\end{subfigure}
		\hfill
		\begin{subfigure}[t]{.3\linewidth}
			\centering \includegraphics[page=2]{planar-correctness-chord}
			\caption{$f$ is incident to $w$.}
			\label{fig:planar-correctness-chord-case2}
		\end{subfigure}
		\hfill
		\begin{subfigure}[t]{.3\linewidth}
			\centering \includegraphics[page=3]{planar-correctness-chord}
			\caption{$f$ is incident to $u$.}
			\label{fig:planar-correctness-chord-case3}
		\end{subfigure}
		\caption{Cases when there are chords in~$G_i'$. We always find a good vertex.}
		\label{fig:planar-correctness-chord}
	\end{figure}
	
	Otherwise all faces in $\hat F$ are incident either to $u$ or to $w$.
	Note that none of the faces in $\hat F$ can be incident to both $u$
	and $w$ as this would imply 
    a triangle.
	Assume that there is a face $f \in \hat F$ incident to a half-chord
	with endpoints $w$ and $w'$; see \cref{fig:planar-correctness-chord-case2}.
	We claim that we can pick a free vertex $v$ of $f$.
	Since at least $u$ stays invisible, $v$ observes rule~\ref{rule:invisible}.
	Rules~\ref{rule:chord} and~\ref{rule:half-chord} are trivially observed.
	Rule~\ref{rule:neighbor} is also observed, since $w'$ is not incident
	to a chord whose other endpoint is visible (as they are all inside $\hat G_i'$)
	and $w$ is not incident to a chord whose other endpoint is visible,
	otherwise invariant~\ref{inv:edges} would be violated.
	
	Finally, all faces in $\hat F$ are incident to~$u$.
	For a face $f \in \hat F$, let $u'$ be the other endpoint
	of the half-chord; see \cref{fig:planar-correctness-chord-case3}.
	Note that it might happen now that all free vertices of~$f$
	break rule~\ref{rule:invisible} or rule~\ref{rule:neighbor}.
	We claim, however, that we can always pick~$u'$.
	Because of $u$, $u'$ observes rule~\ref{rule:invisible}.
	Furthermore, $u'$ observes rule~\ref{rule:chord}
	and rule~\ref{rule:half-chord}
	as $u'$ cannot be incident to a chord or another half-chord distinct from the half-chord $\{u,u'\}$.
	If~$u'$ was incident to a chord or a half-chord,
	this chord or half-chord would divide~$\hat G_i'$ into
	two parts and then there would be a face in $\hat F$ that is not incident to $u$
	due to \cref{clm:bothSidesF}.
	By the same argument, all chords and half-chords
	of $\hat G_i'$ are incident to~$u$.
	In order to show that $u'$ observes rule~\ref{rule:neighbor},
	note that the only vertex of $\hat G_i'$ that is incident to
	a chord whose other vertex is visible is $u$. Since $u'$ is
	adjacent neither to $u$ (otherwise there is a triangle) nor
	to any vertex that is not in $\hat G_i'$, $u'$ observes
	rule~\ref{rule:neighbor} as well and, hence, is a good vertex.
	
	We have shown that there is always a good vertex
	on the outer face of~$G_i'$ if the outer face of~$G_i'$
	is a simple cycle.
	Now assume that the outer face of $G_i'$ is not just a simple cycle, but
	consists of one or multiple cactus graphs.
	If we have multiple cactus graphs, we can consider them individually.
	So, it suffices to consider the case where the outer face of $G_i'$
	is one (connected) cactus graph.
	Still, $H_i$ may be disconnected.
	Let $C_1, C_2, \dots$ be the connected components of~$H_i$, and
	let $\tilde G_1, \tilde G_2, \dots$ be the corresponding subgraphs of $G_i'$.
	Two subgraphs $\tilde G_j$ and $\tilde G_k$
	may be connected by at most one common vertex or via a single edge.
	Otherwise, we consider them as non-connected
	(if they are connected by a path of length $\ge 2$ in $G_i$,
	they are independent because the neighborhood of $\tilde G_j$ does not overlap $\tilde G_k$ and vice versa; these parts remain as a forest in the end).
	Let~$T$ be a graph with a vertex for each $\tilde G_1, \tilde G_2, \dots$ where
	two vertices are adjacent if and only if the corresponding subgraphs are connected.
	Since the outer face of $G_i'$ is a cactus graph, $T$ is a forest.
	Consider the subgraph $\tilde G_1$ and use the algorithm above to
	find a good vertex $v$. If $v$ is a cut vertex, then check if
	it is also a good vertex in all subgraphs from $\{\tilde G_1, \tilde G_2, \dots\}$
	where it is contained as well.
	Further, check for each neighbor~$w$ of $v$ that is contained in a
	subgraph~$\tilde G_j$ distinct from~$\tilde G_1$ whether making $w$ visible
	violates one of the invariants (note that this is a weaker criterion than
	checking if $w$ is a good vertex and it implies that $w$
	and its neighbors in~$\tilde G_j$ are not good vertices).
	If there is a subgraph~$\tilde G_j$
	where picking $v$ breaks at least one rules
	(or making a neighbor of~$v$ visible breaks an invariant),
	then find a good vertex in $\tilde G_j$
	(recall that there exists at least one good vertex) and proceed in the same way.
	Since $T$ does not contain cycles, this procedure always terminates
	with a (globally) good vertex.
	
        Our constructive proof can be turned into a polynomial-time algorithm.
        Let $n$ be the number of vertices of~$G$.

    \begin{claim}
      \label[claim]{clm:runtime}
      The algorithm runs in $O(n^2)$ time.
    \end{claim}
	
    \begin{claimproof}
    If $\Gamma$ is not given, we can compute it in
    linear time~\cite{fpp-hdpgg-Combin90,s-epgg-SODA90}.
    We mark the vertices that are currently on the outer face.
    This can be done initially in linear time.
    Updating the vertices on the outer face can be done in linear
    total time as every vertex disappears at most once.
    To keep track when a vertex disappears,
    we maintain, for each vertex,
    a counter of the neighbors that have not been visible yet.
    Whenever a vertex appears, we update this counter over the incident edges.
    As we have a linear number of edges and at most two updates per edge,
    this can also be done in linear total time.
    
    We next describe how to represent $G'_i$.
    For each of the two sides of every edge,
    we keep a flag that is true if and only if
    the side is incident to the outer face.
    The flags can be initialized in linear time.
    Whenever a vertex disappears,
    we traverse the faces that
    have become part of the outer face
    and update the corresponding flags; as every
    edge side becomes incident to the
    outer face at most once,
    this step can be done in linear total time.
    Using these flags allows us to
    keep track of the vertices on the outer face of~$G'_i$
    and the inner vertices (half-chord vertices) of~$G'_i$.
    We maintain a list~$L$ of the vertices on the outer face of~$G'_i$.
    At the same time, for each vertex, we maintain a mark
    that is true if the vertex lies on the outer face of~$G'_i$.
    A vertex~$v$ lies on the outer face of~$G'_i$
    if $v$ is on the current outer face (according to its mark)
    and if at least one of the edges incident to~$v$
    has a flag saying that one of its sides is not part of the outer
    face.
    The list~$L$ and the corresponding flags can be 
    initialized for~$G'_1$ in linear time.
    Then, we can mark in linear time the inner vertices of~$G'_1$
    by checking, for every vertex~$v$
    that is not marked to be on the
    current outer face,
    whether $v$ has two or more neighbors
    on the outer face of~$G'_1$.
    We update $L$ and all flags marking vertices of~$G'_i$
    whenever a vertex disappears.
    When updating, we also maintain,
    for each inner vertex of~$G'_i$,
    how many of its neighbors on the outer face of~$G'_i$ are visible.
    All of this can be done in linear time
    per iteration, and in quadratic total time.
    Note that we do not explicitly mark the edges of $G'_i$.
    For each edge~$e$, we can check in constant time
    whether both endpoints of~$e$ belong to~$G'_i$
    and at least one of them lies on the outer face~--
    then $e$ belongs to~$G'_i$.
    
    Finally, we discuss the running time for finding a good vertex,
    which we pick in the next iteration.  By definition, a vertex is
    good if it observes
    rules~\ref{rule:invisible}--\ref{rule:half-chord}.  We can simply
    iterate over each vertex~$v$ in~$L$, and check in $\deg(v)$ time
    whether it observes all rules.  In particular, note that, for the
    inner half-chord vertices neighboring~$v$ (relevant for
    rule~\ref{rule:half-chord}), we have a counter saying how many of
    its neighbors on the outer face of~$G'_i$ are visible (and we know
    whether $v$ is visible itself).  We have shown that there always
    exists a good vertex in~$L$.  In the worst case, these checks (and
    hence each iteration) take linear total time.  Since the number of
    iterations is linear, the algorithm runs in quadratic total
    time.
  \end{claimproof}
  This finishes the proof of the theorem.
\end{proof}

\section{Open Problems}
\label{sec:open}

\begin{enumerate}
	\item It is interesting to study further parameterizations of \forestStory and
	\outerStory. Is there an \FPT\ algorithm with respect to the
	treewidth (pathwidth) of the given graph?
	\item While we extended the existing planar storyplan problem into the
	direction of less powerful but easier-to-understand storyplans, one
	could also go into the opposite direction and investigate more
	powerful storyplans in order to be able to construct such storyplans
	for larger classes of graphs.
	For example, 1-planar storyplans would be a natural direction
        for future research.  It turns out, however, that 
        deciding whether a given graph admits a 1-planar
        storyplan is also \NP-hard~\cite{s-1ps-BTh24}.
        Hence, it would be interesting to see whether the
        problem admits a parameterized approach for this class of
        graphs.
      \item What are the smallest graphs in terms of the number of
        vertices or edges that do not admit forest storyplans,
        outerplanar storyplans, planar storyplans, or 1-planar
        storyplans?
\end{enumerate}

\bibliographystyle{plainurl}
\bibliography{abbrv,storyplan} 

\end{document}